\documentclass[11pt]{article}
\usepackage[margin=1in]{geometry}
\usepackage{amsthm,amsfonts,amsmath,amssymb,mathtools}
\usepackage{graphicx}
\usepackage{xcolor}
\usepackage[normalem]{ulem}
\usepackage{url}
\usepackage{booktabs}
\usepackage{multirow}
\usepackage{enumitem}
\usepackage{float}
\usepackage{tikz}
\usepackage{thmtools}
\usepackage{thm-restate}
\usepackage[bookmarks=false,breaklinks=false,pdfborder={0 0 1},
            colorlinks=true,allcolors=blue]{hyperref}
\usepackage[nameinlink]{cleveref}

\newtheorem{theorem}{Theorem}
\newtheorem{lemma}[theorem]{Lemma}
\newtheorem{proposition}[theorem]{Proposition}
\newtheorem{corollary}[theorem]{Corollary}
\newtheorem{observation}[theorem]{Observation}
\theoremstyle{definition}
\newtheorem{definition}[theorem]{Definition}
\newtheorem{remark}[theorem]{Remark}

\newcommand{\mcA}{\mathcal{A}}
\newcommand{\mcB}{\mathcal{B}}
\newcommand{\E}{\operatorname{\mathbb{E}}}
\renewcommand{\Pr}{\operatorname{Pr}}
\newcommand{\Ind}{\mathbf{1}}
\newcommand{\Z}{\mathbb{Z}}

\def\opt{\mathrm{OPT}}
\def\cost{\mathrm{Cost}}
\def\fwd{\mathrm{fwd}}
\def\medalg{\mathrm{med}}
\def\refalg{\mathrm{ref}}
\def\cmpalg{\mathrm{cmp}}
\def\blkalg{\mathrm{blk}}
\def\latealg{\mathrm{late}}
\newcommand{\nxt}{\operatorname{nxt}}
\newcommand{\lst}{\operatorname{lst}}
\newcommand{\thr}{\operatorname{thr}}
\newcommand{\md}{\operatorname{med}}

\usetikzlibrary{arrows.meta,calc,decorations.pathreplacing,positioning}
\tikzset{
  pg/.style   = {circle, draw, line width=0.7pt, minimum size=7.5mm,
                 inner sep=0pt, font=\small},
  gap/.style  = {pg, draw=red!65!black, fill=red!8, dashed, line width=0.9pt},
  here/.style = {pg, draw=blue!60!black, fill=blue!8, line width=0.9pt},
  hop/.style = {-{Stealth[length=2mm,width=1.6mm]}, line width=0.7pt},
}

\author{Ali Vakilian\thanks{Rutgers University. Email: \texttt{a.vakilian@rutgers.edu}. The author was supported in part by NSF CAREER Award CCF-2542925.}}
\title{Tight Lower Bounds for Distributional Paging}
\date{}

\begin{document}
\maketitle

\begin{abstract}
Distributional paging asks for a polynomial-time eviction rule that competes with the optimal \emph{online} policy $\opt$ when requests come from a known stochastic source. Two algorithms proposed by Lund, Phillips, and Reingold [FOCS, 1994] are canonical. Their randomized \emph{dominating distribution} algorithm, designed for Markov sources, draws the page to evict from any distribution over the cache obeying a family of pairwise linear inequalities; their deterministic \emph{median} algorithm, designed for sources whose inter-arrival times are independent across pages, discards the cached page with the largest lower median of its conditional next-request time. The original guarantees were $4$ and $5$; Pabbaraju and Vakilian [ICALP 2025] recently sharpened them to $2$ and $4$.

\smallskip

We show that neither sharpened constant admits further improvement, in its respective source model. For Markov paging we give two complementary factor-$2$ lower-bound constructions, each with an explicit polynomial-time dominating selector. A directed-cycle family has ratio $2k/(k+1)$, while a recurrent four-page family with cache size three has ratio $2-4\varepsilon/3$.

Our second shows that the median algorithm fails to be $c$-competitive whenever $c<4$ in the independent inter-arrival model. These are separate results for distinct stochastic models. Since each lower bound matches the corresponding upper bound, the two competitive ratios are $2$ and $4$ in their respective settings. The strongest bounds previously were $1.5907$ and $1.511$.

\smallskip
The two lower bounds use different mechanisms. For the dominating distribution algorithm, the directed cycle exposes the half-mass mechanism in its cleanest form: half of the eviction mass is forced onto the page returning last, while a bad selector places the remaining half on the page returning first. The four-page construction turns the same extreme-point choice into a recurrent cascade across two slowly switching blocks. 
For the median algorithm we build independent renewal processes that repeatedly mislead the median rule: the page with the latest median is nevertheless likely to be requested early, forcing a cascade of bad evictions, while a simple online policy can be made to pay arbitrarily close to one quarter as much.
\end{abstract}

\section{Introduction}
\label{sec:intro}

A paging system keeps $k$ of $n$ pages in a fast cache. Requests arrive one at a time; a request for a cached page is free, while a request for an absent page triggers a \emph{fault}, costs one unit, and obliges the system to throw out a resident page to make room for the newcomer. The design question is which page to throw out.

Classical competitive analysis, following Sleator and Tarjan~\cite{sleator1985amortized}, measures an online rule against the best \emph{offline} schedule on a worst-case sequence, and the resulting verdict is bleak: deterministic rules cannot beat a ratio of $k$~\cite{sleator1985amortized}, and randomized ones cannot beat $H_k=\Omega(\log k)$, as Fiat, Karp, Luby, McGeoch, Sleator, and Young~\cite{fiat1991competitive} showed. Actual workloads are of course not adversarial; they concentrate their requests, a phenomenon Denning~\cite{denning1968working} named locality of reference, and a long line of models has tried to capture it~\cite{borodin1991competitive, irani1992strongly, fiat1995randomized, koutsoupias2000beyond, becchetti2004modeling}.


\emph{Distributional} paging proceeds differently. It supposes the requests are produced by a source whose law the algorithm knows, and it swaps the offline yardstick for $\opt$, the best \emph{online} policy; one that knows the source but must still guess at the future. That swap is what puts constant factors within reach. Two sources are standard. Karlin, Phillips, and Raghavan~\cite{karlin1992markov} introduced \emph{Markov paging}, in which a known Markov chain emits the requests; they observed that the optimal limiting-average policy is obtained from a linear program with $n\binom nk$ state variables, while a finite-horizon optimum can be computed by a dynamic program with exponential dependence on $k$, which is why one asks instead for polynomial-time rules. Lund, Phillips, and Reingold~\cite{lund1999paging} studied the \emph{independent inter-arrival} source, where every page runs its own renewal process and distinct pages are independent.

The two source models are distinct. A superposition of independent renewal processes is generally not first-order Markov on the page labels, and a Markov request chain need not decompose into independent renewal processes. We therefore keep the models separate throughout. In particular, the latter makes no claim about the median algorithm on Markov sources. 

\medskip
Lund et al.~\cite{lund1999paging} proposed a different eviction rule for each of the two models: the dominating distribution algorithm in their more general distributional framework, specialized here to Markov paging, and the median algorithm for independent inter-arrival sources. For the Markov model, given pages $p,q$ and a page $s$ requested most recently, put
\begin{align}
    \label{eqn:alpha-def}
    \alpha(p<q \mid s) := \Pr\nolimits_M\left[\,p \text{ returns before } q \;\middle|\; s \text{ requested most recently}\,\right],
\end{align}
adopting $\alpha(p<p\mid s)=0$. Call a distribution $\mu$ on the cache $C$ \emph{dominating} when
\begin{align}
    \label{eqn:dominating}
    \E_{p \sim \mu}\left[\alpha(p<q \mid s)\right] \le \frac12 \qquad \text{for all } q \in C .
\end{align}
Equivalently, for every resident $q$, a page $p\sim\mu$ is requested before $q$ with probability at most $1/2$. The \emph{dominating distribution algorithm} evicts $p\sim\mu$ at each fault; a dominating $\mu$ is always available from a linear program of size $O(k)$~\cite{borodin2005online, roughgardennotes}. 

In the independent inter-arrival model, called the \emph{IP-paging model} by Lund et al.~\cite{lund1999paging}, the source specifies an inter-arrival distribution for every page, and successive inter-arrival times are independent across pages and renewal intervals. From this distribution and the elapsed time since a page’s most recent request, the algorithm obtains the conditional distribution of its residual waiting time. The \emph{median algorithm} computes a median of this residual-time distribution for every resident page and evicts a page with the largest such median.
Lund et al.~\cite{lund1999paging} proved approximation ratios of 4 for the dominating-distribution algorithm and 5 for the median algorithm, the latter under the additional assumption of independent inter-arrival times. The factor-$4$ bound stood for three decades, although a factor of $2$ had subsequently been suggested as plausible~\cite{roughgardennotes}. Pabbaraju and Vakilian~\cite{pabbaraju2025new} tightened the charging argument underlying both analyses and reduced the guarantees to $2$ and $4$, respectively. 
Our lower bounds show that these guarantees are tight: neither can be improved.

\paragraph{Our contribution.}
Our contribution is the two matching lower bounds. The first needs a word of framing. Since \eqref{eqn:dominating} is a system of inequalities, the dominating distribution algorithm is not one algorithm but a class, and on a given instance that class is usually large. It can even contain excellent members: on the very instances we build, one dominating $\mu$ is a point mass reproducing an optimal online policy. Thus these directed-cycle instances cannot give a lower bound that holds uniformly for every dominating selector, and any lower bound extracted from them must name a member. The upper bound guarantee of Pabbaraju and Vakilian~\cite{pabbaraju2025new} is proved for every dominating $\mu$ at once, so to show that its constant is best possible as a uniform guarantee, it suffices to produce one dominating selector whose ratio approaches $2$.

\begin{restatable}{theorem}{thmdomlb}
    \label{thm:dom-lb}
    There is an explicit dominating distribution algorithm, computable in polynomial time, that is not $c$-competitive against $\opt$ in the Markov paging model for any $c<2$.
\end{restatable}

We prove \Cref{thm:dom-lb} twice. The first proof uses the forward-biased selector on directed cycles and obtains the exact ratio $2k/(k+1)$. The second uses a different explicit selector on four pages with cache size three and obtains $2-4\varepsilon/3$ for every $0<\varepsilon<1/2$. The four-page construction is stronger as a fixed-size lower bound and has a positive stationary fault rate on every fixed instance; the directed cycle remains the shorter and simpler proof.

Because the lower-median algorithm is deterministic and our construction has a unique maximizing page at every fault, neither a selector nor an eviction tie-breaking issue arises. Its lower bound applies to the independent inter-arrival model with known deterministic initial ages, rather than to Markov paging.

\begin{restatable}{theorem}{thmmedlb}
    \label{thm:med-lb}
    In the independent inter-arrival model, allowing known deterministic initial ages, the median algorithm is not $c$-competitive against $\opt$ for any $c<4$.
\end{restatable}

The theorems live in different source models, Markov paging and independent renewals, respectively, and each meets the best known guarantee for its own.

\begin{corollary}
    \label{cor:exact}
    In their respective source models, the competitive ratio of the algorithm of \Cref{thm:dom-lb} equals $2$, and the competitive ratio of the median algorithm equals $4$ under the initialization convention of \Cref{sec:prelim}.
\end{corollary}
\begin{proof}
    The algorithm of \Cref{thm:dom-lb} is a dominating distribution algorithm, hence $2$-competitive by~\cite[Theorem 1]{pabbaraju2025new}, while \Cref{thm:dom-lb} rules out anything smaller. The median algorithm is $4$-competitive by~\cite[Theorem 2]{pabbaraju2025new} under the initial-age extension, while \Cref{thm:med-lb} rules out anything smaller.
\end{proof}


\begin{table}[H]
    \centering
    \setlength{\tabcolsep}{5pt}
    \renewcommand{\arraystretch}{1.15}
    \begin{tabular}{@{}lllcc@{}}
    \toprule
    Source & Algorithm & Reference & Upper & Lower \\
    \midrule
    \multirow{3}{*}{Markov}
      & \multirow{3}{*}{\begin{tabular}{@{}l@{}}Dominating distribution\\(randomized)\end{tabular}}
      & Lund et al.~\cite{lund1999paging}               & $4$ & --- \\
      & & Pabbaraju and Vakilian~\cite{pabbaraju2025new} & $\mathbf{2}$ & $1.5907$ \\
      & & \textbf{This work} (\Cref{thm:dom-lb}) & --- & $\mathbf{2}$ \\
    \midrule
    \multirow{3}{*}{\begin{tabular}{@{}l@{}}Independent\\inter-arrival\end{tabular}}
      & \multirow{3}{*}{\begin{tabular}{@{}l@{}} Median\\(deterministic)\end{tabular}}
      & Lund et al.~\cite{lund1999paging}               & $5$ & $1.511$\footnotemark \\
      & & Pabbaraju and Vakilian~\cite{pabbaraju2025new} & $\mathbf{4}$ & --- \\
      & & \textbf{This work} (\Cref{thm:med-lb})        & --- & $\mathbf{4}$ \\
    \bottomrule
    \end{tabular}
    \caption{Competitive ratios against the optimal online policy $\opt$. The two rows concern different source models and are not cross-model comparisons.  
    }
    \label{table:summary}
\end{table}
\footnotetext{A proof of $1.4$ appears in~\cite{lund1999paging}; the stronger figure $1.511$ is reported there without proof.}

\subsection{Technical Overview}
\label{sec:techniques}

The Markov and renewal lower bounds address different obstacles, and the Markov result itself has two complementary witnesses. For the dominating distribution algorithm, the eviction rule is not uniquely specified: any point in a feasible polytope may be chosen. We must therefore define one fixed selector and construct instances on which that selector chooses a bad feasible point. The median algorithm has the opposite difficulty: it is deterministic, so the instance must force its unique choice to be bad repeatedly. The common idea is to make each decision rule explicit. On the lazy cycle, all relevant hitting probabilities are $0$ or $1$. 
On the four-page chain, every relevant hitting-order probability is one of $0$, $1$, $\varepsilon$, and $1-\varepsilon$. By symmetry, the joint request/cache process can then be consolidated into six states, whose stationary distribution gives the selector’s exact long-run fault rate.
In the interdependent renewal model's lower bound construction, separated time windows make the conditional medians and the realized request order easy to compare.

\paragraph{Fixing the order in which pages return.}
The quantities in \eqref{eqn:alpha-def} are hitting probabilities and can be difficult to express exactly. Pabbaraju and Vakilian~\cite{pabbaraju2025new} handle this difficulty by restricting their lower-bound search to $n=3$ pages and cache size $k=2$, where the feasible region can be written explicitly and optimized numerically. We instead choose instances on which every relevant hitting probability is trivial. On a deterministic directed cycle, after a request to $s$ the other pages return in the fixed cyclic order $s+1,s+2,\dots,s-1$. Working with the base cycle removes that irrelevant waiting-time layer; \Cref{prop:uniform-lazification} later makes the construction aperiodic without changing the hitting orders or the ratio. Consequently every $\alpha$ is $0$ or $1$, and the constraints in \eqref{eqn:dominating} can be solved directly.

Set $n = k + 1$, so exactly one page is missing from the cache; we call it the \emph{hole}. After a fault on $s$, let $e=s+1$ be the resident requested first and let $\ell=s-1$ be the resident requested last. The constraint in \eqref{eqn:dominating} for $q=\ell$ forces $\mu(\ell)\ge\frac12$. Once this holds, the remaining constraints impose no further restriction on how the other half of the mass is distributed. Our selector maximizes $\mu(e)$, so it assigns all remaining mass to $e$. Hence the unique selected distribution is $\mu=\frac12\delta_e+\frac12\delta_\ell$; see \Cref{lem:split}. The factor of $2$ comes from this split: at least half of the probability must be assigned to the best available page, but the other half may be assigned to the worst available page. The selector is one fixed algorithm maximizing $\mu(s+1\bmod n)$ from the page labels alone, and the cycle is numbered so that this label names the worst resident.

The fault rates follow immediately from the hole picture. Discarding $e$ leaves the hole one request ahead, whereas discarding $\ell$ leaves it $k$ requests ahead. Thus the selector's inter-fault spacings are independent and equal to $1$ or $k$, each with probability $1/2$. The reference policy always discards $\ell$. This is a farthest-in-future choice on every realized trajectory, so the policy is optimal both offline and online (\Cref{lem:ref-optimal}), and its inter-fault spacing is always $k$. The ratio of the two fault rates is therefore $k/((k+1)/2)$, which tends to $2$.

\smallskip
The four-page construction keeps the same bad-vertex-versus-late-page comparison but removes the growing-cache limit. Two two-page blocks alternate internally and switch only with probability $\varepsilon$. At an entry-page miss, the bad vertex again splits mass equally between the next and late pages. At an exit-page miss, the next page loses the hitting-order race with probability $\varepsilon$, so the same tight domination constraint permits the selector to put the slightly larger mass $1/[2(1-\varepsilon)]$ on it. The resulting miss can recreate another state of the same type. Lumping the request/cache process into six states gives stationary fault rates $\varepsilon(1-2\varepsilon/3)$ for the bad selector and $\varepsilon/2$ for the policy that always evicts the late page, hence ratio $2-4\varepsilon/3$.

\paragraph{Repeated faults under independent renewals.}
Directed cycles do not hurt the median algorithm: the page farthest ahead in the cycle also has the latest median next-request epoch, so the median rule makes the optimal eviction. Our second construction instead uses the independent renewal model defined in \Cref{sec:prelim}. It consists of $k+1$ pages and a cache of size $k$, so exactly one page is absent from the cache at any time. It creates a deliberate mismatch between the ordering of the conditional medians and the realized request order: a page can have the latest median and still be requested early with substantial probability.
The programming lemma (\Cref{lem:program}) lets us prescribe the first request of each renewal process. Block $j$ contains five ordered windows $S_j \prec G_j \prec T_j \prec F_j \prec L_j$. 
A live {\em anchor} page requests in $S_j$ with probability $a=\frac12-\gamma$ and in $L_j$ with probability $b=\frac12+\gamma$. Its lower median is in $L_j$, after every competing median, so the median policy repeatedly evicts anchors until the selected anchor takes its late branch. The number of anchor faults in a block is geometric with success parameter $b$ and hence has mean $1/b$.

The comparison policy evicts one {\em token} page instead. Conditional on still being active in block $j$, that token requests at the checkpoint $T_j$ with probability exactly $1/2$; this is the comparison policy's only possible fault in the block. A one-pair cache invariant (\Cref{lem:invariant}) gives expected block costs $1/b$ and $1/2$. Their ratio approaches $2/b = 4/(1 + 2\gamma)$, which can be made arbitrarily close to $4$. Finally, \Cref{lem:finite} replaces the idealized infinite pools by finite ones and converts the construction to a horizon measured by the number of requests.

\subsection{Related Work}
\label{sec:related-work}

Much work has sought to escape the $\Omega(\log k)$ randomized barrier~\cite{fiat1991competitive}. Access graphs~\cite{borodin1991competitive, irani1992strongly, fiat1995randomized} constrain which page may follow which; diffuse adversaries~\cite{koutsoupias2000beyond, young1998bounding} and adversarially biased inputs~\cite{young2000line} confine the adversary to distributions near a nominal one; working-set and locality parameterizations~\cite{denning1968working, torng1998unified, becchetti2004modeling, dorrigiv2005survey, angelopoulos2009paging} grade instances by how concentrated they are. Markov paging~\cite{karlin1992markov} belongs to this tradition and reads naturally as a randomized access graph. What sets it and the independent inter-arrival model apart is the yardstick: both compete with the best online policy, which is what admits constant factors in the first place. Irani~\cite{irani2005competitive} and Borodin and El-Yaniv~\cite{borodin2005online} survey the area.

\paragraph{Algorithms for distributional paging.}
Karlin, Phillips, and Raghavan~\cite{karlin1992markov} found that intuitive heuristics ``evicting on largest \emph{expected} return time, for one'' can be off by $\Omega(k)$, and gave the first constant-factor rule, though with an unwieldy constant. Lund, Phillips, and Reingold~\cite{lund1999paging} then contributed the dominating distribution and median algorithms, with ratios $4$ and $5$. A single charging argument underlies both, and it concedes a factor of $2$ for pages charged twice over, on top of whatever the rule itself costs. Pabbaraju and Vakilian~\cite{pabbaraju2025new} eliminated the first factor to obtain $2$ and $4$; their analysis tolerates approximate values of \eqref{eqn:alpha-def}, which matters when the chain must be learned from data~\cite{hao2018learning, wolfer2021statistical}. Our results complement theirs by showing that neither constant can be lowered in the corresponding source model.

\paragraph{Lower bounds.}
Progress here has lagged far behind. For the median algorithm the previously proved record was $1.4$, proved by Lund et al.~\cite{lund1999paging}, who report $1.511$ without proof. For the dominating distribution algorithm the record was the $1.5907$ that Pabbaraju and Vakilian~\cite{pabbaraju2025new} extracted by searching a three-page instance. \Cref{thm:dom-lb,thm:med-lb} raise these to the optimal $2$ and $4$. What the constructions add beyond the numbers is a reason: the mass split at the trailing page in one case, the compelled anchor cascade in the other, rather than a constant certified by search.

\section{Preliminaries}
\label{sec:prelim}

\paragraph{Paging and the yardstick.}
Write $V$ for the page set, $|V|=n$, and let requests $X_1,X_2,\dots$ issue from a known source. A policy holds $k<n$ pages; it pays a unit at step $t$ when $X_t$ is absent, and must then discard a resident to install $X_t$. For a request horizon $T\in\mathbb{N}$, meaning the first $T$ requests of the sequence, define
\begin{align}
    \label{eqn:objective}
    \E[\cost(\mcA,T)] := \E\left[\sum_{t=1}^T \Ind\left[\mcA \text{ faults at step } t\right]\right],
\end{align}
averaging over the source together with any coins the policy tosses. For each fixed horizon $T$, let $\opt$ (or $\opt_T$ when the dependence matters) minimize \eqref{eqn:objective}; it is an online policy, informed about the source and ignorant of the future. Say $\mcA$ is \emph{$c$-competitive} when $\E[\cost(\mcA,T)]\le c\cdot\E[\cost(\opt,T)]$ on every instance and horizon, and let its \emph{competitive ratio} be the least such $c$. Both of our proofs rest on the same elementary move, worth isolating once.

\begin{observation}
    \label{obs:online-witness}
    Let $\mcB$ be any online policy. Then $\E[\cost(\opt,T)]\le\E[\cost(\mcB,T)]$ for every $T$, since $\opt$ minimizes \eqref{eqn:objective} over online policies. So to certify that some $\mcA$ is not $c$-competitive it is enough to exhibit an instance, a horizon, and an online $\mcB$ with $\E[\cost(\mcA,T)]>c\cdot\E[\cost(\mcB,T)]$.
\end{observation}

\paragraph{Markov sources and dominating distributions.}
A Markov paging instance~\cite{karlin1992markov} fixes a transition matrix $M\in[0,1]^{n\times n}$, known to the algorithm, whose entry $M_{i,j}$ gives the chance that $j$ follows $i$. Under the standard rational-input convention, the hitting probabilities $\alpha(p<q\mid s)$ needed in \eqref{eqn:alpha-def} are computable in polynomial time by finite-chain hitting-probability linear systems, and $\mu$ is dominating when \eqref{eqn:dominating} holds. We start the chain uniformly and give the algorithm and $\opt$ a common initial cache.

\begin{theorem}[Lund, Phillips, and Reingold~{\cite[Theorem 3.4]{lund1999paging}}]
    \label{thm:dom-exists}
    On every instance and at every fault a dominating distribution exists, and one may be found by a linear program with $O(k)$ variables and $O(k)$ constraints.
\end{theorem}

\paragraph{Independent inter-arrival sources and the median rule.}
For the median result we use the following explicit independent-inter-arrival convention. Each page $v$ has a known positive gap distribution $X_v$, and its successive inter-arrival times are independent copies of $X_v$. The renewal processes of distinct pages are mutually independent. At the beginning of the instance, page $v$ may have a deterministic age $a_v$, known to the policy; equivalently, its first residual waiting time has distribution $X_v-a_v \ \big|\  X_v>a_v$.
After a request to $v$, a fresh independent copy of $X_v$ determines its next request. Merging the renewal epochs of all pages gives the request sequence. Policies observe both request identities and their global request times. Hence the current vector of page ages is part of the observed history, so the policy can update every conditional residual-time distribution.

No tie-breaking between simultaneous requests is needed in our construction. In \Cref{sec:timing}, we assign every page a distinct residue modulo an integer $D$ and choose every possible inter-arrival time to be a multiple of $D$. All requests to a page therefore remain in its assigned residue class, and two distinct pages can never request at the same time. 

This source model is not Markov paging: after independent renewal processes are merged, the law of the next requested page generally depends on the entire vector of current ages, not only on the most recently requested page. Moreover, our lower bound uses the deterministic initial ages above. It therefore does not, as written, imply the same lower bound for fresh age-zero initialization, stationary initialization, or Markov paging.

At a fault at global time $t$, let $W_v(t)$ be the conditional residual waiting time of a resident page $v$. Because these distributions may be discrete, the word ``median'' need not specify a unique number. We make the algorithm deterministic by using the lower $1/2$-quantile, referred to as the {\em lower median},
\begin{align}
    \label{eqn:lower-median}
    \md(Y) := \inf\left\{x : \Pr[Y \le x] \ge \tfrac12\right\}.
\end{align}
The median algorithm $\mcA_{\medalg}$ evicts a resident maximizing $\md(W_v(t))$. Equivalently, it maximizes the median next-request epoch $t+\md(W_v(t))$; the latter notation is more convenient in \Cref{sec:med-lb}, where medians are located in named time windows. The lower-median convention selects the earliest median when a discrete distribution has several medians. In our finite-support instances it is computed directly from the listed atoms. We additionally arrange that the maximizing page is unique at every fault, so no eviction tie-breaking is used. 

\section{A Tight Lower Bound for the Dominating Distribution Algorithm}
\label{sec:dom-lb}

This section establishes \Cref{thm:dom-lb} in two ways. Sections~\ref{sec:selector}--\ref{sec:dom-wrapup} retain the  directed-cycle proof: they define the forward-biased selector, compute its fault rate and the optimal rate, and obtain $2k/(k+1)$. Section~\ref{sec:four-page} then gives a fixed-size four-page refinement with ratio $2-4\varepsilon/3$. The common aperiodicity argument is provided at the end in \Cref{prop:uniform-lazification}.

\subsection{The Forward-Biased Selector}
\label{sec:selector}

An algorithm must choose which dominating $\mu$ it picks. Ours uses the chain only through the domination constraints; its bias among feasible distributions is determined by the fixed page labels.

\begin{definition}[Forward-biased selector]
    \label{def:selector}
    Number the pages $0,1,\dots,n-1$ once and for all. Let a request to $s$ fault, with $C$ the cache just beforehand, and put $e:=s+1 \bmod n$. The rule $\mcA_\fwd$ evicts a draw from an optimal solution of
    \begin{align}
        \label{eqn:selector-lp}
        \max\ \mu(e) \quad \text{subject to} \quad
        \sum_{p \in C}\mu(p)\,\alpha(p<q\mid s) \le \tfrac12 \ \ \text{for every } q \in C,
        \qquad \mu \in \Delta(C),
    \end{align}
    where $\Delta(C)$ is the simplex on $C$, and we extend $\mu(p):=0$ for every $p\notin C$. Among all optimal solutions, choose the lexicographically largest vector $(\mu(0),\mu(1),\dots,\mu(n-1))$.
\end{definition}

\begin{proposition}
    \label{prop:selector-ok}
    On every Markov paging instance $\mcA_\fwd$ is a dominating distribution algorithm, and each of its evictions is computable in time polynomial in $n$, $k$, and the input encoding length.
\end{proposition}
\begin{proof}
By \Cref{thm:dom-exists}, the feasible region of
\eqref{eqn:selector-lp} is nonempty. It is also a closed subset of the
simplex $\Delta(C)$ and is therefore compact. Consequently, the objective
$\mu(e)$ attains its maximum.

The lexicographically largest maximizer is well defined. It can be obtained
by a sequence of linear programs. First solve \eqref{eqn:selector-lp} and
let $v_e$ be the optimal value of $\mu(e)$. Add the equality
$\mu(e)=v_e$. Then process the pages in their fixed numerical order. For
each coordinate not already fixed, maximize that coordinate subject to the
original constraints and all equalities introduced in the preceding
steps, and fix it at its optimal value. After at most $k$ additional
linear programs, every coordinate in $C$ has been fixed, yielding the
unique lexicographically largest optimum $\mu^\star$.

Every feasible solution of \eqref{eqn:selector-lp} satisfies
\eqref{eqn:dominating}. In particular, $\mu^\star$ is a dominating
distribution, so $\mcA_\fwd$ is a dominating-distribution algorithm.

Each of the linear programs has $k$ variables and $O(k)$ constraints,
including the nonnegativity, normalization, domination, and previously
fixed-coordinate constraints. Assuming that the transition probabilities
and the values $\alpha(p<q\mid s)$ are rational and have polynomial
encoding length, standard linear-programming algorithms solve every one
of these programs in time polynomial in $n$, $k$, and the total input encoding length. Since only $O(k)$ programs are
solved, $\mu^\star$ is computable in polynomial time. Its coordinates are
rational numbers of polynomial encoding length, so a page can be sampled
exactly from $\mu^\star$ in expected polynomial time.
\end{proof}

\Cref{def:selector} describes one algorithm, run unchanged on every input. Where the numbering bears no relation to the dynamics its objective is meaningless and $\mcA_\fwd$ may do perfectly well. The instances below are numbered to make it do as badly as possible.

\subsection{Directed Cycles}
\label{sec:cycle}

Given $k\ge2$, put $n=k+1$ and name the pages by $\Z_n$. We work directly with the deterministic directed cycle as follows,
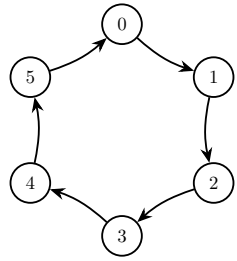
\begin{figure}[H]
    \centering
    \begin{minipage}[c]{0.50\textwidth}
        \centering
        \def\figscale{0.7}
        \begin{tikzpicture}[scale=\figscale, transform shape]
            \def\N{6}
            \def\R{2}
            \foreach \i in {0,...,5} {
                \pgfmathsetmacro{\ang}{90-360*\i/\N}
                \node[pg] (p\i) at (\ang:\R) {$\i$};
            }
            \foreach \i
                [evaluate=\i as \j using {int(mod(\i+1,\N))}]
                in {0,...,5} {
                \draw[hop] (p\i) to[bend right=12] (p\j);
            }
        \end{tikzpicture}
    \end{minipage}\hfill
    \begin{minipage}[c]{0.46\textwidth}
        \centering
        \begin{equation*}
            \label{eqn:cycle}
            \begin{aligned}
                M_{i,\,i+1\bmod n} &= 1,\\
                M_{i,j} &= 0
                \qquad \text{for } j\neq i+1\bmod n.
            \end{aligned}
        \end{equation*}
    \end{minipage}
    \caption{The directed cycle
    with $n=6$.}
    \label{fig:cycle}
\end{figure}

The chain is irreducible, has period $n$, and has the uniform stationary law. We use the periodic base chain because it makes the proof entirely deterministic. If an aperiodic source is desired, \Cref{prop:uniform-lazification} adds uniform self-loops without changing any hitting-order probability or the competitive-ratio calculation. Let the cache begin as $\Z_n\setminus\{0\}$.

With $n=k+1$ the cache always lacks exactly one page, which we call the \emph{hole}. Faults happen exactly when the hole is requested, and servicing one relocates the hole to whichever page is discarded.
\begin{observation}
    \label{obs:order}
    Suppose $s$ was requested most recently. The next requests are $s+1$, then $s+2$, and so on. Hence the members of $\Z_n\setminus\{s\}$ reappear in the predetermined cyclic sequence
    \begin{align}
        \label{eqn:order}
        e = s+1, \quad s+2, \quad \dots, \quad \ell = s-1 \pmod n,
    \end{align}
    so that for distinct $p,q\in\Z_n\setminus\{s\}$,
    \begin{align}
        \label{eqn:alpha-cycle}
        \alpha(p<q\mid s) = \begin{cases} 1, & p \text{ comes before } q \text{ in \eqref{eqn:order}},\\ 0, & \text{otherwise}.
        \end{cases}
    \end{align}
    In particular $\alpha(e<q\mid s)=1$ for every $q\ne e$, while $\alpha(\ell<q\mid s)=0$ for every $q$.
\end{observation}

Under the numbering in directed cycles, the successor $e=s+1$ selected in \Cref{def:selector} is the resident page requested next, and hence the most costly one to evict.

\begin{lemma}
    \label{lem:split}
    On directed cycles, at every fault on a page $s$, program \eqref{eqn:selector-lp} is solved uniquely by
    \begin{align}
        \label{eqn:split}
        \mu = \tfrac12\delta_e + \tfrac12\delta_\ell, \qquad e = s+1,\quad \ell = s-1 \pmod n .
    \end{align}
    So $\mcA_\fwd$ discards $e$ or $\ell$, each with probability $\tfrac12$.
\end{lemma}
\begin{proof}
    A fault on $s$ means the hole is $s$, so $C=\Z_n\setminus\{s\}$; as $n\ge3$ we have $e,\ell\in C$ with $e\ne\ell$. Thus, \eqref{eqn:split} is dominating: fix $q\in C$ and read off \eqref{eqn:alpha-cycle}. For $q=e$ both $\alpha(e<e\mid s)$ and $\alpha(\ell<e\mid s)$ vanish, the first by convention and the second because $\ell$ is last in \eqref{eqn:order}, so the left side of \eqref{eqn:dominating} is $0$. For $q\ne e$ we get $\alpha(e<q\mid s)=1$, as $e$ leads \eqref{eqn:order}, and $\alpha(\ell<q\mid s)=0$, by convention if $q=\ell$ and because $\ell$ trails otherwise; the left side is $\tfrac12$. Either way \eqref{eqn:dominating} holds, at objective value $\tfrac12$.

    Next, let $\nu$ be dominating and test \eqref{eqn:dominating} at $q=\ell$. Every $p\in C$ other than $\ell$ precedes $\ell$ in \eqref{eqn:order}, so $\alpha(p<\ell\mid s)=1$, whereas $\alpha(\ell<\ell\mid s)=0$; the constraint therefore reads
    \begin{align}
        \label{eqn:half-on-ell}
        1-\nu(\ell) = \sum_{p \in C,\ p \ne \ell}\nu(p)
        = \sum_{p \in C}\nu(p)\,\alpha(p<\ell\mid s) \le \tfrac12 .
    \end{align}
    Thus $\nu(\ell)\ge\tfrac12$ and $\nu(e)\le\tfrac12$. Equality $\nu(e)=\tfrac12$ forces $\nu(\ell)=\tfrac12$ and exhausts all mass, so \eqref{eqn:split} is the unique optimum.
\end{proof}

\Cref{lem:split} contains the factor of $2$ in its simplest form. The constraint \eqref{eqn:half-on-ell}, indexed by the trailing page $\ell$, forces at least half of the mass onto the best available choice. The constraints leave the remaining half unrestricted, and our selector assigns it to the leading page $e$.

\subsection{Fault Rates and the Optimal Policy}
\label{sec:rates}

\Cref{fig:outcomes} shows the two possible inter-fault spacings.

\begin{figure}[H]
    \centering
    \begin{tikzpicture}
        \def\N{6}
        \def\R{1.6}
        \begin{scope}[xshift=-3.1cm]
            \foreach \i/\sty/\lab in {0/here/{$s$}, 1/gap/{$e$}, 2/pg/{},
                                      3/pg/{}, 4/pg/{}, 5/pg/{$\ell$}} {
                \pgfmathsetmacro{\ang}{90-360*\i/\N}
                \node[\sty] (q\i) at (\ang:\R) {\lab};
            }
            \foreach \i [evaluate=\i as \j using {int(mod(\i+1,\N))}] in {0,...,5} {
                \draw[hop, gray!40] (q\i) to[bend right=12] (q\j);
            }
            \draw[hop, red!70!black, line width=1.1pt] (q0) to[bend right=12] (q1);
            \node[align=center, font=\footnotesize] at (0,0) {discard $e$\\w.p.\ $\tfrac12$};
            \node[align=center, font=\scriptsize] at (0,-2.75) {$D=1$ request};
        \end{scope}
        \begin{scope}[xshift=3.1cm]
            \foreach \i/\sty/\lab in {0/here/{$s$}, 1/pg/{$e$}, 2/pg/{},
                                      3/pg/{}, 4/pg/{}, 5/gap/{$\ell$}} {
                \pgfmathsetmacro{\ang}{90-360*\i/\N}
                \node[\sty] (r\i) at (\ang:\R) {\lab};
            }
            \draw[hop, gray!40] (r5) to[bend right=12] (r0);
            \foreach \i [evaluate=\i as \j using {int(mod(\i+1,\N))}] in {0,...,4} {
                \draw[hop, red!70!black, line width=1.1pt] (r\i) to[bend right=12] (r\j);
            }
            \node[align=center, font=\footnotesize] at (0,0) {discard $\ell$\\w.p.\ $\tfrac12$};
            \node[align=center, font=\scriptsize] at (0,-2.75) {$D=k$ requests};
        \end{scope}
    \end{tikzpicture}
    \caption{What one eviction buys, after a fault on $s$ (blue). The hole is the dashed red circle, and the red arcs show the requests before the next fault. Leaving the hole at the leading page $e$ causes a fault one request later; leaving it at the trailing page $\ell$ causes a fault $k$ requests later.}
    \label{fig:outcomes}
\end{figure}
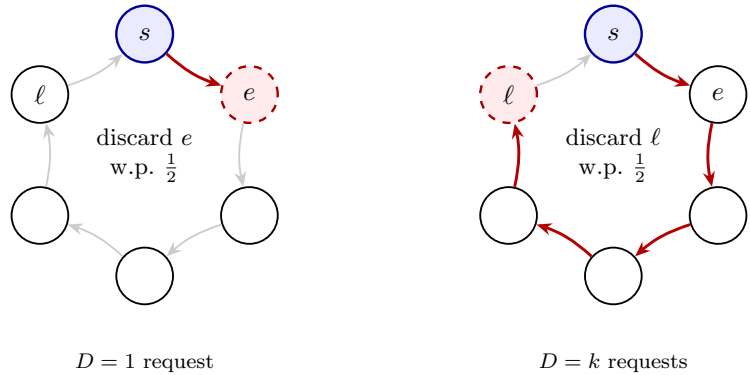

Let $\mcA_\refalg$ discard $\ell=s-1\bmod n$ at every fault on $s$.

\begin{lemma}
    \label{lem:ref-optimal}
    On directed cycles, $\E[\cost(\mcA_\refalg,T)] = \E[\cost(\opt,T)]$ for every $T$.
\end{lemma}
\begin{proof}
    Fix a realized request stream and a fault on $s$. By \Cref{obs:order}, $\ell=s-1$ is the last resident to return, so discarding $\ell$ is a farthest-in-future eviction. Belady's rule~\cite{belady1966study} is offline optimal, hence $\mcA_\refalg$ matches the offline optimum on every trajectory. Since $\mcA_\refalg$ is also online, \Cref{obs:online-witness} gives equality with $\opt$ in expectation at every horizon.
\end{proof}

The following elementary bounded-renewal lemma converts the
inter-fault spacing distributions of the two policies into their
asymptotic expected fault rates.
\begin{lemma}[Bounded renewal rate]
    \label{lem:rate}
    Fix $k\ge2$ and consider a possibly randomized policy $\mcA$ on the
    directed-cycle instance. Let $J$ be the index of its
    first fault, and suppose that $1\le J\le k+1$. After the first fault,
    let $D_1,D_2,\ldots$ be the numbers of requests between consecutive
    faults. Suppose that they are i.i.d.\ copies of a random variable $D$
    satisfying $1\le D\le k$, and that $J$ is independent of
    $D_1,D_2,\ldots$. Then, uniformly over $T\ge1$,
    \begin{align}
        \E[\cost(\mcA,T)]
        =\frac{T}{\E[D]}+O_k(1).
    \end{align}
\end{lemma}

\begin{proof}
    Write $\bar d:=\E[D]$, and define
    \begin{align*}
        S_0:=0,\qquad
        S_r:=\sum_{i=1}^r D_i,\qquad
        R(m):=\max\{r\ge0:S_r\le m\}.
    \end{align*}
    Thus $R(m)$ counts the complete inter-fault spacings contained in
    $m$ requests after a fault. Put $\tau:=R(m)+1$. Since
    $1\le D_i\le k$, $m<S_\tau\le m+k$ and $\tau\le m+1$. 
    Moreover, $\{\tau\ge i\}=\{S_{i-1}\le m\}$ depends only on
    $D_1,\ldots,D_{i-1}$ and is independent of $D_i$. 
    Since $\tau\le m+1$, we may write $S_\tau = \sum_{i=1}^{\tau}D_i = \sum_{i=1}^{m+1}D_i\Ind[\tau\ge i]$.
    Taking expectations and using linearity gives
    \begin{align}
        \E[S_\tau]
        =\sum_{i=1}^{m+1}
          \E\!\left[D_i\Ind[\tau\ge i]\right].
        \label{eqn:stopped-sum-expansion}
    \end{align}
    Since for each $i$, the event $\{\tau\ge i\}$ is independent of $D_i$, $\E\!\left[D_i\Ind[\tau\ge i]\right] = \E[D_i]\Pr(\tau\ge i) = \bar d\,\Pr(\tau\ge i)$, where $\bar d = \E[D]$. 
    Substituting this identity into
    \eqref{eqn:stopped-sum-expansion}, and using the tail-sum formula
    for the positive integer-valued random variable $\tau$, yields $\E[S_\tau] =\bar d\sum_{i=1}^{m+1}\Pr(\tau\ge i) =\bar d\,\E[\tau]$.
    Combining $m<\E[S_\tau]\le m+k$ with $\E[S_\tau]=\bar d\,\E[\tau]$ gives ${m}/{\bar d} <\E[\tau] \le (m+k)/{\bar d}$.
    Finally, since $\tau=R(m)+1$,
    \begin{align}
        \frac{m}{\bar d}-1
        <\E[R(m)]
        \le\frac{m+k}{\bar d}-1.
        \label{eqn:renewal-count-bounds}
    \end{align}
    Because $1\le\bar d\le k$, the difference between
    $\E[R(m)]$ and $m/\bar d$ is bounded by a constant depending only
    on $k$, uniformly over $m$. Hence $\E[R(m)] ={m}/{\bar d}+O_k(1)$.
    Conditional on $J=j$, the number of faults among the first $T$
    requests is $0$ when $T<j$, and is $1+R(T-j)$ when $T\ge j$. Since $j\le k+1$, $\E[R(m)] ={m}/{\bar d}+O_k(1)$ implies,
    uniformly over $T$ and $j$, $\E[\cost(\mcA,T)\mid J=j] = {T}/{\bar d} + O_k(1)$. Averaging over $J$ proves the claim.
\end{proof}

\begin{proposition}
    \label{prop:fwd-rate}
    On the directed cycle instance, $\E[\cost(\mcA_\fwd,T)] = \frac{2}{k+1}T + O_k(1)$.
\end{proposition}
\begin{proof}
    After a fault, \Cref{lem:split} leaves the hole at $e$ or $\ell$, independently with probability $1/2$. Thus the next fault is respectively $1$ or $k$ requests later. The inter-fault spacings are therefore i.i.d. with
    \[
        \Pr[D=1]=\Pr[D=k]=\tfrac12,
        \qquad \E[D]=\frac{k+1}{2}.
    \]
    Applying \Cref{lem:rate} with $\E[D]=(k+1)/2$ gives bound on the expected cost of $\mcA_\fwd$.
\end{proof}

\begin{proposition}
    \label{prop:opt-rate}
    On the directed cycle instance, $\E[\cost(\opt,T)] = \frac{1}{k}T + O_k(1)$.
\end{proposition}
\begin{proof}
    By \Cref{lem:ref-optimal} it suffices to analyze $\mcA_\refalg$. After every fault it leaves the hole at $\ell$, which is requested exactly $k$ requests later. Hence every inter-fault spacing equals $k$. Applying \Cref{lem:rate} with $D=k$ deterministically proves the cost analysis.
\end{proof}

\subsection{Proof of \texorpdfstring{\Cref{thm:dom-lb}}{Theorem 1}}
\label{sec:dom-wrapup}

\thmdomlb*

\begin{proof}
    \Cref{prop:selector-ok} makes $\mcA_\fwd$ a dominating, polynomial-time rule. Fix $k\ge2$ and run it on \eqref{eqn:cycle}. Dividing \Cref{prop:fwd-rate} by \Cref{prop:opt-rate}, the bounded additive terms are swamped by the linear growth:
    \begin{align}
        \label{eqn:limit}
        \lim_{T\to\infty}\frac{\E[\cost(\mcA_\fwd,T)]}{\E[\cost(\opt,T)]}
        = \frac{2/(k+1)}{1/k} = \frac{2k}{k+1} .
    \end{align}
    Given $c<2$, pick an integer $k>\frac{c}{2-c}$. Then \eqref{eqn:limit} exceeds $c$, so $\mcA_\fwd$ is not $c$-competitive. If aperiodicity is imposed as part of the Markov-source convention, apply the uniform lazification of \Cref{prop:uniform-lazification}; it preserves both the selector's domination feasibility and the ratio in \eqref{eqn:limit}.
\end{proof}

\subsection{A Fixed-Size Four-Page Construction}
\label{sec:four-page}
The directed-cycle proof approaches $2$ by increasing the cache size to infinity. We now give a second proof that keeps $n=4$ and $k=3$ fixed. The constant is again approached through a family of instances, now by sending a transition probability $\varepsilon$ to zero. The two proofs share the same local geometry: a tight domination constraint leaves a bad extreme point supported on a soon-to-be-requested page and a safe late page. The difference is how the bad choice reproduces itself. The cycle does so after every fault by rotation; the four-page chain does so recurrently through rare switches between two blocks.

\begin{theorem}[Four-page refinement]
    \label{thm:four-page-lb}
    There is a globally defined dominating distribution algorithm $\mcA_\blkalg$, computable in polynomial time, with competitive ratio exactly $2$. More precisely, for every $0<\varepsilon<1/2$ there is an irreducible Markov paging instance $I_\varepsilon$ on four pages with cache size three such that
    \begin{align}
        \label{eqn:block-main-ratio}
        \liminf_{T\to\infty}
        \frac{\E[\cost(\mcA_\blkalg,T)]}{\E[\cost(\opt,T)]}
        \ge 2-\frac{4}{3}\varepsilon .
    \end{align}
    Both the witness selector and a comparison policy have positive linear fault rates on every fixed $I_\varepsilon$.
\end{theorem}

\paragraph{The instance.}
Fix a rational $\varepsilon\in(0,\tfrac12)$, let $V=\{0,1,2,3\}$, and split $V$ into the two \emph{blocks} $\{0,1\}$ and $\{2,3\}$. Pages $0$ and $2$ are \emph{entry} pages and pages $1$ and $3$ are \emph{exit} pages. The nonzero transition probabilities are
\begin{align}
    \label{eqn:blocks}
    M_{0,1}=M_{2,3}=1, \qquad
    M_{1,0}=M_{3,2}=1-\varepsilon, \qquad
    M_{1,2}=M_{3,0}=\varepsilon;
\end{align} 
see \Cref{fig:blocks}.
Thus requests alternate inside the current block until an exit page switches to the other block's entry page, an event of probability $\varepsilon$. The positive directed cycle $0\to1\to2\to3\to0$ makes the chain irreducible. \Cref{prop:uniform-lazification} makes it aperiodic without changing the ratio.

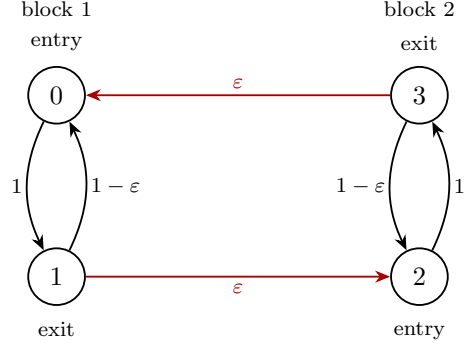
\begin{figure}[H]
    \centering
    \begin{tikzpicture}[scale=1.0, transform shape]
        \node[pg] (p0) at (0,1.2)   {$0$};
        \node[pg] (p1) at (0,-1.2)  {$1$};
        \node[pg] (p2) at (4.8,-1.2){$2$};
        \node[pg] (p3) at (4.8,1.2) {$3$};

        \draw[hop] (p0) to[bend right=26]
            node[midway, left,  font=\scriptsize, inner sep=1.5pt] {$1$} (p1);
        \draw[hop] (p1) to[bend right=26]
            node[midway, right, font=\scriptsize, inner sep=1.5pt] {$1-\varepsilon$} (p0);
        \draw[hop] (p2) to[bend right=26]
            node[midway, right, font=\scriptsize, inner sep=1.5pt] {$1$} (p3);
        \draw[hop] (p3) to[bend right=26]
            node[midway, left,  font=\scriptsize, inner sep=1.5pt] {$1-\varepsilon$} (p2);

        \draw[hop, red!65!black] (p1) --
            node[midway, below, font=\scriptsize, inner sep=2.5pt] {$\varepsilon$} (p2);
        \draw[hop, red!65!black] (p3) --
            node[midway, above, font=\scriptsize, inner sep=2.5pt] {$\varepsilon$} (p0);

        \node[font=\scriptsize, anchor=south] at (0,1.66)    {entry};
        \node[font=\scriptsize, anchor=north] at (0,-1.66)   {exit};
        \node[font=\scriptsize, anchor=south] at (4.8,1.66)  {exit};
        \node[font=\scriptsize, anchor=north] at (4.8,-1.66) {entry};
        \node[font=\scriptsize] at (0,2.34)   {block $1$};
        \node[font=\scriptsize] at (4.8,2.34) {block $2$};
    \end{tikzpicture}
    \caption{The two-block chain \eqref{eqn:blocks}.}
    \label{fig:blocks}
\end{figure}
Since $|V|=k+1$, exactly one page is missing after every request. At a miss on $r$, the candidates are $C_r=V\setminus\{r\}$. Define
\begin{align}
    \label{eqn:roles}
    \renewcommand{\arraystretch}{1.1}
    \begin{array}{c|cccc}
     r&0&1&2&3\\ \hline
     \nxt(r)&1&0&3&2\\
     \lst(r)&3&3&1&1\\
     \thr(r)&2&2&0&0
    \end{array}
    \qquad\qquad
    a_r := \begin{cases}
        1, & r \text{ an entry page},\\
        1-\varepsilon, & r \text{ an exit page}.
    \end{cases}
\end{align}
Here $\nxt(r)$ is the other page in $r$'s block, while $\thr(r)$ and $\lst(r)$ are respectively the entry and exit pages of the other block.
\begin{lemma}[Hitting-order table]
    \label{lem:block-alpha}
    For every $r\in V$,
    \begin{align*}
        \alpha(\nxt(r)<\thr(r)\mid r)
        &=\alpha(\nxt(r)<\lst(r)\mid r)=a_r,\\
        \alpha(\thr(r)<\nxt(r)\mid r)
        &=\alpha(\lst(r)<\nxt(r)\mid r)=1-a_r,\\
        \alpha(\thr(r)<\lst(r)\mid r)&=1,
        \qquad\alpha(\lst(r)<\thr(r)\mid r)=0.
    \end{align*}
\end{lemma}
\begin{proof}
    The involution $(0\ 2)(1\ 3)$ is an automorphism, so it is enough to consider $r=0$ and $r=1$. From $r=0$, the next request is $1=\nxt(0)$, so it precedes pages $2$ and $3$ surely. Moreover, page $3$ can first be entered only from page $2$, and hence $2=\thr(0)$ precedes $3=\lst(0)$ surely.

    From $r=1$, the next request is $0=\nxt(1)$ with probability $1-\varepsilon$, in which case it precedes both pages in the other block. With probability $\varepsilon$, the next request is $2=\thr(1)$ and is followed deterministically by $3=\lst(1)$; both then precede page $0$. Again $2$ precedes $3$ surely. These cases give the table.
\end{proof}
At a miss on $r$, abbreviate the masses on $\nxt(r),\lst(r),\thr(r)$ by $\mu_n,\mu_\ell,\mu_t$.

\begin{lemma}
    \label{lem:block-polytope}
    For every miss state $r$, the domination constraints are equivalent to
    \begin{align}
        \label{eqn:block-polytope}
        \mu_n,\mu_\ell,\mu_t&\ge0, &
        \mu_n+\mu_\ell+\mu_t&=1, &
        a_r\mu_n+\mu_t&\le\frac12.
    \end{align}
    The vertices are $\mathbf e_{\lst(r)}$, $\tfrac12\mathbf e_{\lst(r)}+\tfrac12\mathbf e_{\thr(r)}$, and $\mu_r^*:=\beta_r\mathbf e_{\nxt(r)}+(1-\beta_r)\mathbf e_{\lst(r)}$ where $\beta_r:=\frac{1}{2a_r}$.  
\end{lemma}
\begin{proof}
    Substituting \Cref{lem:block-alpha} into the constraints indexed by $q=\nxt(r),\lst(r),\thr(r)$ gives, respectively,
    \[
        (\mu_\ell+\mu_t)(1-a_r)\le\tfrac12,
        \qquad a_r\mu_n+\mu_t\le\tfrac12,
        \qquad a_r\mu_n\le\tfrac12.
    \]
    The first inequality is automatic because $1-a_r\le\varepsilon<1/2$, and the third follows from the second. This proves \eqref{eqn:block-polytope}. Since $a_r>1/2$, intersecting the probability triangle with $a_r\mu_n+\mu_t\le1/2$ leaves exactly the displayed three vertices.
\end{proof}

We compare two online paging
policies here. The \emph{block selector} $\mcA_\blkalg$ is the dominating
selector whose fault rate we lower-bound, while the \emph{late policy}
$\mcA_\latealg$ is a comparison policy with a smaller fault rate.
Whenever a request to page $r$ causes a miss, the block selector
$\mcA_\blkalg$ evicts a resident page according to the distinguished
extreme point $\mu_r^*$ from \Cref{lem:block-polytope}. Here $r$
denotes the page requested at the current miss, so it varies over time.
There is one distribution $\mu_r^*$ for each of the four possible
current pages. The name ``block selector'' refers to the two-block
structure of the Markov chain. By
\Cref{prop:selector-ok,lem:block-polytope}, $\mcA_\blkalg$ is a
polynomial-time dominating selector.
More explicitly, $\mcA_\blkalg$ evicts the page requested immediately next with probability
\begin{align}
    \label{eqn:block-betas}
    \beta_E&=\frac12
    &&\text{when the miss occurs on an entry page},\\
    \beta_X&=\frac{1}{2(1-\varepsilon)}
    &&\text{when the miss occurs on an exit page},
\end{align}
and otherwise evicts $\lst(r)$.

The late policy $\mcA_\latealg$ always evicts $\lst(r)$, the resident page whose next request occurs latest given that the current request is to $r$. Thus $\mcA_\latealg$ avoids the immediately-next page whenever possible. The name ``late policy'' refers to this latest-next-request eviction rule.

\begin{figure}[H]
    \centering
    \begin{tikzpicture}[
        scale=.95, transform shape,
        pg/.append style={draw=black,text=black},
        here/.append style={draw=blue!60!black,text=black},
        gap/.append style={draw=red!65!black,text=black},
        hop/.append style={draw=black},
        every node/.append style={text=black}]
        \def\R{1.55}
        \begin{scope}[xshift=-3.25cm]
            \node[font=\footnotesize] at (0,2.18) {entry-page miss};
            \node[here] (er) at (90:\R) {$r$};
            \node[gap]  (en) at (0:\R) {$\nxt$};
            \node[pg]   (et) at (-90:\R) {$\thr$};
            \node[gap]  (el) at (180:\R) {$\lst$};
            \draw[hop,red!70!black,line width=1.05pt] (er) to[bend right=12]
                node[above right,font=\scriptsize] {$1$} (en);
            \draw[hop,red!70!black,line width=1.05pt] (en) to[bend right=12] (et);
            \draw[hop,red!70!black,line width=1.05pt] (et) to[bend right=12] (el);
            \node[align=center,font=\scriptsize] at (0,-2.22)
                {evict $\nxt$ w.p. $\beta_E=\tfrac12$; evict $\lst$ otherwise};
        \end{scope}

        \begin{scope}[xshift=5.25cm]
            \node[font=\footnotesize] at (0,2.18) {exit-page miss};
            \node[here] (xr) at (90:\R) {$r$};
            \node[gap]  (xn) at (0:\R) {$\nxt$};
            \node[pg]   (xt) at (-90:\R) {$\thr$};
            \node[gap]  (xl) at (180:\R) {$\lst$};
            \draw[hop,red!70!black,line width=1.05pt] (xr) to[bend right=12]
                node[above right,font=\scriptsize] {$1-\varepsilon$} (xn);
            \draw[hop] (xr) to
                node[right,font=\scriptsize,fill=white,inner sep=1pt] {$\varepsilon$} (xt);
            \draw[hop] (xt) to[bend right=12]
                node[below left,font=\scriptsize] {$1$} (xl);
            \node[align=center,font=\scriptsize] at (0,-2.22)
                {evict $\nxt$ w.p. $\beta_X=\tfrac{1}{2(1-\varepsilon)}$; evict $\lst$ otherwise};
        \end{scope}
    \end{tikzpicture}
    \caption{The local extreme-point choices in the circular hole notation of \Cref{fig:outcomes}. The blue node is the current request and the dashed red nodes are the two possible new holes. 
    }
    \label{fig:block-vertex}
\end{figure}
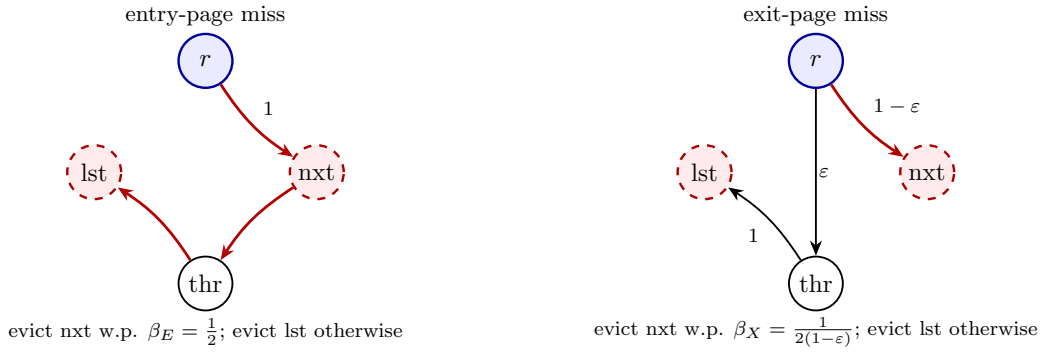

\begin{proposition}
    \label{prop:block-rates}
    On $I_\varepsilon$, under the common initialization convention,
    \begin{align}
        \E[\cost(\mcA_\blkalg,T)]
        =\varepsilon\left(1-\frac23\varepsilon\right)T+o(T),\text{and } \quad
        \E[\cost(\mcA_\latealg,T)]
        =\frac{\varepsilon}{2}T+o(T).
    \end{align}
\end{proposition}
The proof groups the joint request/cache chain into six symmetry states and solves its stationary balance equations. The full transition verification and calculation are given in \Cref{app:block-rates}.

\begin{proof}[Proof of \Cref{thm:four-page-lb}]
    Since $\mcA_\latealg$ is an online policy using the same initial cache, \Cref{obs:online-witness} gives $\E[\cost(\opt,T)]\le\E[\cost(\mcA_\latealg,T)]$ for every $T$. Therefore \Cref{prop:block-rates} implies
    \begin{align*}
        \liminf_{T\to\infty}
        \frac{\E[\cost(\mcA_\blkalg,T)]}{\E[\cost(\opt,T)]}
        \ge
        \lim_{T\to\infty}
        \frac{\E[\cost(\mcA_\blkalg,T)]}{\E[\cost(\mcA_\latealg,T)]}
        =\frac{\varepsilon(1-2\varepsilon/3)}{\varepsilon/2}
        =2-\frac43\varepsilon.
    \end{align*}
    Given $c<2$, choose $0<\varepsilon<\min\{1/2,3(2-c)/4\}$. Then the ratio exceeds $c$, so some finite horizon violates $c$-competitiveness. Conversely, $\mcA_\blkalg$ is a dominating distribution algorithm and is therefore $2$-competitive by~\cite[Theorem 1]{pabbaraju2025new}. Its competitive ratio is exactly $2$.
\end{proof}

The two Markov constructions now admit a precise comparison. They have the same limiting lower-bound constant. The four-page construction is stronger with respect to instance size, since it works with fixed $n=4$ and $k=3$ and derives the gap from recurrent stationary behavior. The directed cycle is stronger with respect to simplicity and selector naturalness: its hitting orders are deterministic, its reference policy is exactly optimal, and the bad selector is the uniform forward-biased LP rule of \Cref{def:selector}, rather than a rule that recognizes a special family.

\subsection{Aperiodicity by Uniform Lazification}

\begin{proposition}[Uniform lazification]
    \label{prop:uniform-lazification}
    Let $M$ be either base chain used in the two Markov lower bounds, and let $M_\lambda:=\lambda I+(1-\lambda)M$ for $0<\lambda<1$. Then every hitting-order probability $\alpha(p<q\mid s)$ and every corresponding domination polytope is unchanged. For the witness and comparison policies used above, every limiting fault rate is multiplied by $1-\lambda$, so their ratio is unchanged. If $M$ is irreducible, then $M_\lambda$ is irreducible and aperiodic.
\end{proposition}
\begin{proof}
    Couple $M_\lambda$ by independently declaring each transition a self-loop with probability $\lambda$ and otherwise taking one transition of $M$. Removing the self-loops leaves exactly an $M$-trajectory. A self-loop repeats the current request, causes a cache hit, and leaves the cache state unchanged. It also cannot change which of two distinct pages is visited first. Thus the $\alpha$-values, the feasible dominating distributions, and the embedded request/cache process at non-self-loop transitions are the same.

    Because the self-loop probability is uniform across pages, the number of genuine $M$-transitions among the first $T$ requests is $(1-\lambda)T+o(T)$ in expectation. Hence each long-run fault rate is multiplied by $1-\lambda$, while any ratio of such rates is preserved. Positive self-loops make every state have period one, and the positive off-diagonal edges of $M$ preserve irreducibility. Applying the operation to the directed cycle and to the four-page chain gives aperiodic versions of both lower-bound families.
\end{proof}

\section{A Tight Lower Bound for the Median Algorithm in Independent Renewal Model}
\label{sec:med-lb}

We now prove \Cref{thm:med-lb}. We first explain why the directed cycle used in \Cref{sec:cycle} cannot also yield a median lower bound. After a fault at page $s$, the resident page $s+j$ is requested exactly $j$ requests later for $j=1,\ldots,k$. Its next-request time is deterministic and hence equals its median. Therefore the median rule evicts $s+k=s-1\pmod{k+1}$, exactly the optimal choice from \Cref{lem:ref-optimal}. The cycle is unsuitable for a second reason: its page-request processes are coupled and do not form independent renewals.
\subsection{Programming the First Request}
\label{sec:program}

The model gives each page an i.i.d. gap law and an initial age, whereas the construction needs a prescribed distribution for the page's first request. The following lemma supplies it and makes the page harmless after that first request.
Throughout, a \emph{gap} is the time between two consecutive requests of a
\emph{single} page.
\begin{lemma}[Programming a first request]
    \label{lem:program}
    Let $\mu$ be a finitely supported distribution on positive times, fix $0<d<A$ and $0<\delta<\frac12$, and define a positive gap variable $X$ by
    \begin{align}
        \label{eqn:program}
        \Pr[X=d]=1-\delta, \qquad
        \Pr[X=A+y]=\delta\,\mu(y) \quad (y\in\operatorname{supp}\mu).
    \end{align}
    Equivalently, $X$ is supported on $\{d\}\cup\{A+y:y\in\operatorname{supp}\mu\}$, each value being a distance from the page's previous request: the short atom $d$, or a long atom $A+y$ placed at offset $y$ beyond the current time. Then,
    \begin{itemize}
        \item A renewal process with gap law $X$ and known initial age $A$ has first residual waiting time distributed as $\mu$,\footnote{Here \emph{initial age $A$} means the process's previous request occurred at time $-A$, so at time $0$ it has already waited for duration $A$ without a request.} and
        \item After its first request, whenever the current post-first gap has age $u<d$, its residual lower median is at most $d-u\le d$. 
    \end{itemize}
\end{lemma}
\begin{proof}
    First note that the formula indeed defines a probability distribution without any further restriction on the support values: since every $y>0$ and $A>d$, each long atom $A+y$ lies strictly above $A$ and is therefore distinct from the short atom $d$, while the long-atom masses sum to $\delta\sum_y\mu(y)=\delta$.
    Because $d<A$ and every support value $y$ is positive, the event $\{X>A\}$ consists exactly of the long atoms $A+y$ and has probability $\delta$ (\Cref{fig:program} visualizes these two operations). Hence, for every $y\in\operatorname{supp}\mu$,
    \[
        \Pr[X-A=y\mid X>A]
        =\frac{\delta\mu(y)}{\delta}
        =\mu(y),
    \]
    which proves the first claim. After the first request, suppose the current gap has age $u<d$. Every support point of $X$ exceeds $u$, so conditioning on $X>u$ leaves probability $1-\delta>1/2$ on the short gap $X=d$. The residual waiting time therefore has an atom of mass greater than $1/2$ at $d-u$, and its lower median is at most $d-u\le d$.
\end{proof}

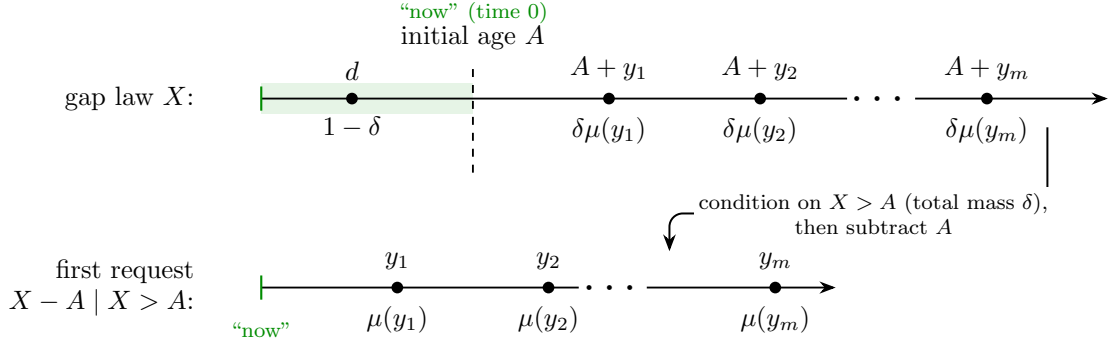
\begin{figure}[H]
\centering
\begin{tikzpicture}[font=\small, x=1cm, y=1cm, >=Stealth,
  nn/.style={color=green!55!black, font=\scriptsize, align=center}]

  \node[anchor=east] at (-0.15,1.6) {gap law $X$:};

  \fill[green!55!black!10] (0.6,1.40) rectangle (3.4,1.80);

  \draw[line width=0.7pt] (0.6,1.6) -- (8.35,1.6);
  \draw[->, line width=0.7pt] (9.25,1.6) -- (11.8,1.6);

  \draw[green!55!black, line width=0.8pt] (0.6,1.76) -- (0.6,1.44);

  \fill (1.8,1.6) circle (2.2pt);
  \node[above=3pt] at (1.8,1.6) {$d$};
  \node[below=3pt] at (1.8,1.6) {$1-\delta$};

  \draw[dashed, line width=0.7pt] (3.4,0.62) -- (3.4,2.08);
  \node[above=1pt] at (3.4,2.08) {initial age $A$};
  \node[nn, anchor=south] at (3.4,2.46) {``now'' (time $0$)};

  \foreach \x/\lab/\mass in {
      5.2/{A+y_1}/{\delta\mu(y_1)},
      7.2/{A+y_2}/{\delta\mu(y_2)},
      10.2/{A+y_m}/{\delta\mu(y_m)}} {
    \fill (\x,1.6) circle (2.2pt);
    \node[above=3pt] at (\x,1.6) {$\lab$};
    \node[below=3pt] at (\x,1.6) {$\mass$};
  }
  \node at (8.8,1.58) {\LARGE $\cdots$};



  \draw[->, rounded corners=5pt, line width=0.75pt]
      (11.0,1.22) -- (11.0,0.08) -- (6.0,0.08) -- (6.0,-0.42);
  \node[align=center, fill=white, inner sep=2pt, font=\scriptsize]
      at (8.7,0.08) {condition on $X>A$ (total mass $\delta$),\\then subtract $A$};

  \node[anchor=east, align=right] at (-0.15,-0.9)
      {first request\\$X-A\mid X>A$:};

  \draw[line width=0.7pt] (0.6,-0.9) -- (4.8,-0.9);
  \draw[->, line width=0.7pt] (5.7,-0.9) -- (8.2,-0.9);

  \draw[green!55!black, line width=0.8pt] (0.6,-0.74) -- (0.6,-1.06);
  \node[nn, anchor=north] at (0.6,-1.22) {``now''};

  \foreach \x/\lab/\mass in {
      2.4/{y_1}/{\mu(y_1)},
      4.4/{y_2}/{\mu(y_2)},
      7.4/{y_m}/{\mu(y_m)}} {
    \fill (\x,-0.9) circle (2.2pt);
    \node[above=3pt] at (\x,-0.9) {$\lab$};
    \node[below=3pt] at (\x,-0.9) {$\mass$};
  }
  \node at (5.26,-0.92) {\LARGE $\cdots$};

\end{tikzpicture}
\caption{Conditioning on the known initial age removes the short atom at $d$ and renormalizes the long atoms; subtracting $A$ leaves exactly the prescribed first-request law $\mu$. For the first request, the gap $X$ is measured from the page's previous (fictitious) request at age $0$; the dashed line marks age $A$, the current time (``time $0$''), so a page enters already $A$ into its gap.}
\label{fig:program}
\end{figure}
We call a page \emph{spent} immediately after its programmed first request. In the finite construction on the ``good'' event, which we describe in~\Cref{lem:finite}, every post-first gap equals $d$. Thus a spent page always has residual lower median at most $d$ before the horizon, so at global time $t$ its median next-request epoch is at most $t+d$. We place all construction windows later than this bound, which prevents spent pages from becoming eviction candidates.

\subsection{Pages, Residue Classes, and Windows}
\label{sec:timing}

This subsection fixes everything about \emph{when} requests may occur; the probability laws that use these times follow in \Cref{sec:anchor-token}.

\paragraph{The pages.}
Fix a block count $H\ge1$ and an integer $L\ge1$. The instance uses a single deterministic \emph{starter page} $z$, \emph{anchor pools} $P_1,\ldots,P_{H+1}$, and \emph{token pools} $Q_1,\ldots,Q_{H+1}$, of sizes
\begin{align}
    \label{eqn:pools}
    |P_r| =2L2^{r-1},
    \qquad
    |Q_r| =L2^{r-1}
    \qquad (1\le r\le H+1).
\end{align}
Pools $P_j,Q_j$ with $j\le H$ drive block $j$, while $P_{H+1}$ and $Q_{H+1}$ only supply boundary pages at the very end of the construction; the roles of all these pages are defined in \Cref{sec:anchor-token}, and the sizes \eqref{eqn:pools} matter only for the finite analysis in \Cref{sec:finite}. In total there are $2L(2^{H+1}-1)$ anchors and $L(2^{H+1}-1)$ tokens, hence
$N := 3L(2^{H+1}-1)+1$
pages including the starter.

\paragraph{Residue classes.}
Set $D:=N+1$; there are exactly $N=D-1$ nonzero residues modulo $D$, one per page. Assign the pages distinct nonzero residues $r_v\pmod D$. Every page other than $z$ 
is programmed by \Cref{lem:program}: choose its first-request law $\mu_v$ on times $y\equiv r_v\pmod D$, choose $A_v>D$ with $A_v\equiv-r_v\pmod D$, and use $d=D$. Then its first request has residue $r_v$ (by the first part in \Cref{lem:program} it occurs at a residual atom $y\equiv r_v$), and every later gap is divisible by $D$. 
Indeed, the short gap equals $D$, and every long gap $A_v+y$ with $y\in\operatorname{supp}\mu_v$ satisfies $A_v+y\equiv-r_v+r_v\equiv0\pmod D$.
Give $z$ deterministic gap $D$ and initial age $D-r_z$, so it first requests at $t_0:=r_z$ and thereafter every $D$ time units. Thus every page remains in its own residue class and distinct pages never request simultaneously. The deterministic starter also supplies the periodic requests used later to convert the time construction to a request-count horizon.

From now on, the construction uses $D$ as the short post-first gap; that is, the parameter $d$ from \Cref{lem:program} is fixed to $D$.
\Cref{fig:residues} depicts this invariant. 
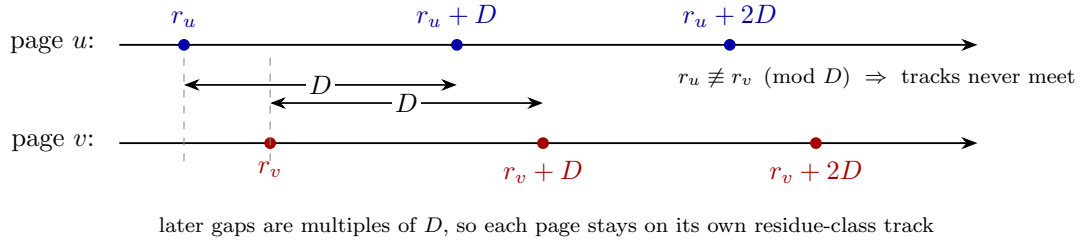
\begin{figure}[H]
\centering
\begin{tikzpicture}[font=\small, x=0.95cm, y=1cm, >=Stealth]
  \node[anchor=east] at (0.45,1.45) {page $u$:};
  \node[anchor=east] at (0.45,0.15) {page $v$:};
  \draw[->, line width=0.7pt] (0.65,1.45) -- (12.6,1.45);
  \draw[->, line width=0.7pt] (0.65,0.15) -- (12.6,0.15);

  \foreach \x/\lab in {1.55/{r_u},5.35/{r_u+D},9.15/{r_u+2D}} {
    \fill[blue!65!black] (\x,1.45) circle (2.2pt);
    \node[above=3pt, text=blue!65!black] at (\x,1.45) {$\lab$};
  }
  \foreach \x/\lab in {2.75/{r_v},6.55/{r_v+D},10.35/{r_v+2D}} {
    \fill[red!65!black] (\x,0.15) circle (2.2pt);
    \node[below=3pt, text=red!65!black] at (\x,0.15) {$\lab$};
  }

  \draw[<->, line width=0.65pt] (1.55,0.92) --
      node[midway, fill=white, inner sep=1pt] {$D$} (5.35,0.92);
  \draw[<->, line width=0.65pt] (2.75,0.68) --
      node[midway, fill=white, inner sep=1pt] {$D$} (6.55,0.68);

  \draw[dashed, black!45] (1.55,1.28) -- (1.55,-0.1);
  \draw[dashed, black!45] (2.75,1.28) -- (2.75,-0.1);
  \node[align=center, font=\scriptsize] at (6.6,-0.95)
      {later gaps are multiples of $D$, so each page stays on its own residue-class track};
  \node[align=center, font=\scriptsize] at (11.2,1.02)
      {$r_u\not\equiv r_v\pmod D$\; $\Rightarrow$\; tracks never meet};
\end{tikzpicture}
\caption{Residue-class separation. The first-request law of each page is supported on one residue class, and every later gap is divisible by $D$. Thus page $u$ always requests at times congruent to $r_u$, page $v$ always requests at times congruent to $r_v$, and distinct residues prevent ties.}
\label{fig:residues}
\end{figure}

\begin{remark}
It is tempting to read the residual identity $X-A_v\sim\mu_v$ as ``$A_v$ becomes $0$ after the first request.'' It does not. The law $\mu_v$ is the \emph{first-request} law only, and it is never redrawn; later gaps are drawn from the gap law $X$ of \Cref{lem:program}, and $A_v$ keeps its fixed value throughout, entering only as the offset inside the long atoms. 
What does reset to $0$ at every request, including the first, is the \emph{renewal age} $u$ (the time since the most recent request); this is the $u$ in the bound $d-u\le d$ of \Cref{lem:program}. In particular the initial age $A_v>D$ is used exactly once to delete the short atom from the first gap, so that the first request follows $\mu_v$; from then on no conditioning is applied.
\end{remark}

\paragraph{Blocks and windows.}
Time is organized into $H$ consecutive \emph{blocks}, and each block $j\in[H]$ contains five disjoint \emph{windows}, in the order
\begin{align}
    \label{eqn:windows}
    S_j \prec G_j \prec T_j \prec F_j \prec L_j,
\end{align}
and $L_j\prec S_{j+1}$ whenever $j<H$. There is no window $S_{H+1}$; after $L_H$ come only the post-horizon atoms. The five windows will contain, respectively, early anchor requests, token filters, a token checkpoint, anchor filters, and late anchor requests; the laws that place atoms there are defined in \Cref{sec:anchor-token}, and \Cref{fig:block} gives a block-level view. For concreteness, every window is an integer interval of length $W:=(N+1)D$, consecutive windows are separated by an empty interval of length $2D$, and the first window is chosen so that $t_0+D<\min S_1$ (recall from the residue-class paragraph that $t_0=r_z$ is the deterministic first-request time of the starter $z$), where $\min S_1$ denotes the first integer time in $S_1$. The exact endpoints are otherwise irrelevant.

\begin{figure}[H]
\centering
\begin{tikzpicture}[
  font=\small, >=Stealth,
  blk/.style={draw, rounded corners=2pt, minimum width=1.95cm,
              minimum height=0.86cm, align=center, font=\scriptsize},
  bdry/.style={draw, dashed, rounded corners=2pt, minimum width=2.35cm,
               minimum height=0.86cm, align=center, font=\scriptsize},
  win/.style={draw, rounded corners=1.5pt, minimum width=1.65cm,
              minimum height=0.82cm, align=center, font=\scriptsize}
]
  \node[font=\bfseries, anchor=east] at (0,2.55) {$H$ blocks:};
  \node[blk] (b1) at (1.30,2.55) {block $1$\\$P_1,Q_1$};
  \node[blk] (b2) at (3.85,2.55) {block $2$\\$P_2,Q_2$};
  \node at (5.70,2.55) {$\cdots$};
  \node[blk] (bH) at (7.55,2.55) {block $H$\\$P_H,Q_H$};
  \node[bdry] (bb) at (10.65,2.55) {boundary only\\$P_{H+1},Q_{H+1}$};
  \draw[->, line width=0.65pt] (b1.east) -- (b2.west);
  \draw[->, line width=0.65pt] (b2.east) -- (5.25,2.55);
  \draw[->, line width=0.65pt] (6.15,2.55) -- (bH.west);
  \draw[->, line width=0.65pt] (bH.east) -- (bb.west);
  \node[align=center, font=\scriptsize] at (5.85,1.82)
      {$P_j$ terminates in block $j$; a $Q_r$ token first reaches a checkpoint in block $r$ and may survive later.};

  \node[font=\bfseries, anchor=east] at (0,0.55) {block $j$:};
  \node[win, fill=red!6]  (S) at (1.15,0.55) {$S_j$\\early anchors};
  \node[win, fill=blue!4] (G) at (3.35,0.55) {$G_j$\\token filters};
  \node[win, fill=blue!8] (T) at (5.55,0.55) {$T_j$\\checkpoint};
  \node[win, fill=gray!8] (F) at (7.75,0.55) {$F_j$\\anchor filters};
  \node[win, fill=red!12] (L) at (9.95,0.55) {$L_j$\\late anchors};
  \draw[->, line width=0.55pt] (S.east) -- (G.west);
  \draw[->, line width=0.55pt] (G.east) -- (T.west);
  \draw[->, line width=0.55pt] (T.east) -- (F.west);
  \draw[->, line width=0.55pt] (F.east) -- (L.west);
  \node[font=\scriptsize] at (5.55,1.18)
      {during $S_j$: median locations are $G_j\prec T_j\prec F_j\prec L_j$};

  \node[font=\scriptsize] at (5.55,-0.22)
      {each box has length $W=(N+1)D$; each arrow represents an empty gap of length $2D$};
\end{tikzpicture}
\caption{Global layout and one-block scale (schematic, not to scale). There are $H$ blocks and no block $H+1$; $P_{H+1},Q_{H+1}$ are boundary pools. Every window has length $W=(N+1)D$, and consecutive windows are separated by an empty interval of length $2D$. }
\label{fig:block}
\end{figure}
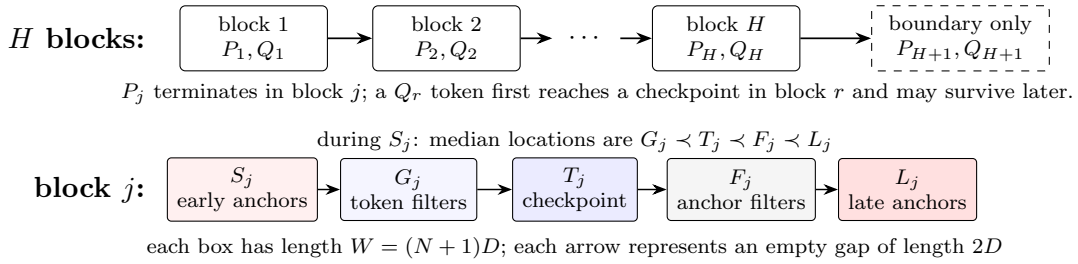

\paragraph{Placing atoms inside a window.}
Cut the window into consecutive \emph{slots} of length $D$; since $W=(N+1)D$,
there is room for $N$ of them. A slot is a run of $D$ consecutive integers (representing time units) and
therefore a complete residue system modulo $D$: it contains exactly one time
congruent to any prescribed residue. So whenever a window must contain atoms of
pages $v_1,\ldots,v_m$ in a prescribed order, where $m\le N$, place the $i$th atom
at the unique time of slot $i$ congruent to $r_{v_i}\pmod D$. Every page keeps its
residue class, and the prescribed order is realized because the slots are consecutive. 
\Cref{fig:placement} depicts one window.
No window is ever asked to hold more than $N$ atoms, since $N$ exceeds the total
number of pages other than the starter.

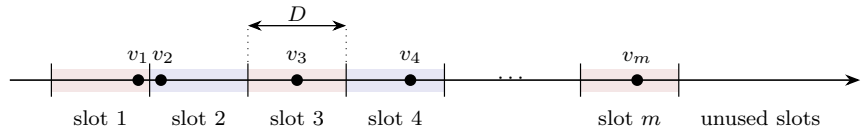
\begin{figure}[H]
\centering
\begingroup
\begin{tikzpicture}[font=\small, x=1cm, >=Stealth,
  lbl/.style={font=\scriptsize, align=center}]

  \foreach \s in {0.3,2.9,7.3}
    \fill[red!55!black!10] (\s,-0.15) rectangle (\s+1.3,0.15);
  \foreach \s in {1.6,4.2}
    \fill[blue!55!black!10] (\s,-0.15) rectangle (\s+1.3,0.15);

  \draw[->, line width=0.6pt] (-0.25,0) -- (11.0,0);
  \foreach \x in {0.3,1.6,2.9,4.2,5.5,7.3,8.6}
    \draw[line width=0.45pt] (\x,0.22) -- (\x,-0.22);

  \fill (1.45,0) circle (2.2pt); \node[above=2.5pt, font=\scriptsize] at (1.45,0) {$v_1$};
  \fill (1.75,0) circle (2.2pt); \node[above=2.5pt, font=\scriptsize] at (1.78,0) {$v_2$};
  \fill (3.55,0) circle (2.2pt); \node[above=2.5pt, font=\scriptsize] at (3.55,0) {$v_3$};
  \fill (5.05,0) circle (2.2pt); \node[above=2.5pt, font=\scriptsize] at (5.05,0) {$v_4$};
  \fill (8.05,0) circle (2.2pt); \node[above=2.5pt, font=\scriptsize] at (8.05,0) {$v_m$};
  \node[font=\scriptsize] at (6.4,0.02) {$\cdots$};

  \draw[<->, line width=0.5pt] (2.9,0.72) -- (4.2,0.72)
      node[midway, above, font=\scriptsize, inner sep=1.5pt] {$D$};
  \draw[dotted, line width=0.4pt] (2.9,0.24) -- (2.9,0.68);
  \draw[dotted, line width=0.4pt] (4.2,0.24) -- (4.2,0.68);

  \node[lbl] at (0.95,-0.5) {slot $1$};
  \node[lbl] at (2.25,-0.5) {slot $2$};
  \node[lbl] at (3.55,-0.5) {slot $3$};
  \node[lbl] at (4.85,-0.5) {slot $4$};
  \node[lbl] at (7.95,-0.5) {slot $m$};
  \node[lbl, anchor=west] at (8.75,-0.5) {unused slots};

\end{tikzpicture}
\endgroup
\caption{A window is cut into consecutive slots of length $D$. Each slot is a run of $D$ consecutive integers, hence a complete residue system modulo $D$, so it contains exactly one time congruent to $r_{v_i}$: put the $i$th atom there. This realizes any prescribed order $v_1,\ldots,v_m$ with $m\le N$ and keeps every page in its own residue class. 
}
\label{fig:placement}
\end{figure}

\subsection{The Anchor and Token Laws}
\label{sec:anchor-token} 
With the timetable of \Cref{sec:timing} in place, we now define the first-request laws. An \emph{anchor} is a page in a pool $P_r$ whose terminal first request is either early, in $S_r$, or late, in $L_r$; these are the pages the median rule is forced to evict. A \emph{token} is a page in a pool $Q_r$ whose first request progresses through token-filter and checkpoint times. Its role in the comparison policy is defined in \Cref{sec:med-dynamics}. An \emph{atom} is only a support point of the random first-request time $Y_v$: each page draws one atom, and that atom is when its first request occurs. A page is \emph{live} until that first request and \emph{spent} afterward. Thus ``anchor'' and ``token'' name page roles, while ``atom'' names a possible request time.

Fix a slack $\gamma\in(0,\tfrac12)$, and write
\begin{align}
    \label{eqn:ab}
    a = \tfrac12-\gamma, \qquad b = \tfrac12+\gamma, \qquad a+b=1.
\end{align}
All probabilities below describe independent first-request times $Y_v$ produced by \Cref{lem:program}, with atoms placed inside their windows, each in its page's residue class, by the scheme of \Cref{sec:timing}. 

\paragraph{Anchors.}
For $1\le r\le H$, an anchor $p\in P_r$ is programmed with the first-request law below.
These probabilities are exactly the prescribed first-request law $\mu_p$ that is fed into \Cref{lem:program}. A simpler equivalent description is sequential: conditional on having survived the previous filters (i.e., $F_1, \ldots, F_{h-1}$), the anchor requests at the current filter $F_h$ with probability $1/2$ and otherwise continues; after surviving the first $r-1$ filters, it chooses the early atom in $S_r$ with probability $a$ and the late atom in $L_r$ with probability $b$. Multiplying these conditional probabilities gives the following unconditional masses.
\begin{align}
    \Pr[Y_p=f_{h,p}] &= 2^{-h} && (1\le h<r), \label{eqn:anchor-filter}\\
    \Pr[Y_p=s_p] &= 2^{-(r-1)}a, \qquad
    \Pr[Y_p=\ell_p] = 2^{-(r-1)}b, \label{eqn:anchor-terminal}
\end{align}
where $f_{h,p}\in F_h$, $s_p\in S_r$, and $\ell_p\in L_r$. Indeed, the filter masses sum to $1-2^{-(r-1)}$, and the terminal pair has total mass $2^{-(r-1)}$, so the law has total mass one. The filters are what keep the construction local in time: in every block $j<r$, a live anchor of $P_r$ has its next unresolved atom in $F_j$, so its median stays inside the current block, below the late window $L_j$ of the block's own anchors (\Cref{lem:ladder}). Without them, an anchor of a later pool would carry the largest median already in block $j$ and the cascade would leave $P_j$. 
Index the anchors in $P_r$ as $p_1,p_2,\ldots$ in the order in which the median rule should expose them, and place their terminal atoms so that
\begin{align}
    \label{eqn:ranks}
    s_{p_1} < s_{p_2} < \cdots,
    \qquad
    \ell_{p_1} > \ell_{p_2} > \cdots .
\end{align}
Both orders are realized by the placement scheme of \Cref{sec:timing}: list the early atoms inside $S_r$ in increasing index order, and the late atoms inside $L_r$ in decreasing index order. 
Concretely, $p_1$ is considered before $p_2$: its early atom satisfies $s_{p_1}<s_{p_2}$, but its late atom satisfies $\ell_{p_1}>\ell_{p_2}$. Hence, if $p_1$ requests early, the median rule next exposes $p_2$; if $p_1$ instead returns late, then $p_2$ has already requested, either in $S_r$ or earlier in $L_r$. The same statement holds for every consecutive pair $p_i,p_{i+1}$.

\paragraph{Boundary anchors.}
A final pool $P_{H+1}$ places mass $2^{-h}$ at a filter atom in $F_h$ for every $h\le H$, and places its remaining mass $2^{-H}$ beyond the horizon. Index this pool as $p_1,p_2,\ldots$ as well, and give smaller-index boundary anchors later post-horizon atoms. This pool supplies the anchor evicted at the final late fault.

\paragraph{Tokens.}
For $1\le r\le H+1$, a token $q\in Q_r$ uses the same half-survival mechanism: conditional on not having requested earlier, it requests at the current filter $G_h$ (for $h<r$) or checkpoint $T_h$ (for $h\ge r$) with probability $1/2$, and otherwise continues. Equivalently, it has
\begin{align}
    \Pr[Y_q=g_{h,q}] &= 2^{-h} && (1\le h<r), \label{eqn:token-filter}\\
    \Pr[Y_q=t_{h,q}] &= 2^{-h} && (r\le h\le H), \label{eqn:token-checkpoint}
\end{align}
where $g_{h,q}\in G_h$ and $t_{h,q}\in T_h$; the remaining mass $2^{-H}$ is placed beyond the horizon. These probabilities also sum to one. The token filters play the same role as the anchor filters: they pin the next unresolved atom of every live token inside the current block, in $G_j$ or $T_j$, so token medians never rise above anchor medians (\Cref{lem:ladder}). For technical reasons explained where the comparison policy is defined (\Cref{sec:med-dynamics}), the token filters must strictly precede the checkpoint within each block.
Call a live token \emph{checkpoint-ready in block $j$} when its next unresolved atom lies in $T_j$. Conditional on the history that made it checkpoint-ready, it requests at $T_j$ with probability $1/2$; if it does not, its next unresolved atom lies in $T_{j+1}$, or at its post-horizon tail when $j=H$.

\paragraph{Ordering the post-horizon atoms.} Place the post-horizon atoms of $P_{H+1}$ strictly after the post-horizon atoms of every token pool, that is, after the tails of every pool $Q_1,\ldots,Q_{H+1}$. The ordinary anchor pools $P_1,\ldots,P_H$ have no post-horizon atoms: all their mass is on filters and their early/late terminal pair. 
A post-horizon atom is an ordinary support point and does carry probability mass ($2^{-H}$ per page), but its \emph{location} is ours to choose, exactly as for atoms inside windows. In \Cref{sec:finite} we cut the run off at a finite request-count horizon and place every post-horizon atom so late that the horizon is provably reached first. The ordering of these atoms therefore influences only median \emph{comparisons}, through \Cref{lem:locate}, and never the realized request sequence.

\subsection{Ordering the Medians}
\label{sec:ladder}

The next two lemmas locate the median next-request epoch of every page and then compare the locations using \eqref{eqn:windows}.

\begin{lemma}
\label{lem:locate}
     Fix any global time $t$ strictly before the programmed first request of a page $p$; equivalently, the page is live at time $t$. Call the smallest atom of $p$'s law exceeding $t$ its \emph{next unresolved atom}.
    \begin{enumerate}[label=(\alph*), itemsep=1pt, topsep=2pt]
        \item If it is an anchor of $P_r$ that has declined all filters $F_1,\ldots,F_{r-1}$, its lower-median next-request epoch is its late atom $\ell_p\in L_r$.
        \item In every other case, its lower-median next-request epoch is its next unresolved atom; concretely, this is its atom in the earliest of its designated windows not yet past $t$, as it has at most one atom per window.
    \end{enumerate}
\end{lemma}
\begin{proof}
    Let $x$ be the page's next unresolved atom. At a filter or checkpoint, the conditional mass at $x$ is exactly $1/2$; if $x$ is the sole post-horizon tail, its conditional mass is $1$. In either case the conditional CDF first reaches $1/2$ at $x$, so the lower median is $x$. For an ordinary anchor page $p$ that has passed all filters, if time $t$ is before its early atom, the remaining conditional masses are $a<1/2$ at $s_p$ and $b>1/2$ at $\ell_p$, so the lower median is $\ell_p$. If $s_p$ has already passed without a request, the late atom $\ell_p$ is the only remaining atom, and the same conclusion holds.
\end{proof}

Throughout the rest of this section we use the following fact about spent pages: at every time considered, each spent page's current post-first gap has age $u<D$, so by the second part of \Cref{lem:program} its residual lower median is at most $D-u\le D$, and its median next-request epoch at global time $t$ is at most $t+D$. This holds because every realized post-first gap equals the short gap $D$: in the finite instance this is exactly the exclusion of the bad event $\mathcal E_3$ in \Cref{lem:finite}, valid before the horizon on its good event, and the idealized experiment adopts it by convention (spent pages request every $D$ time units; the idealized experiment's conventions are collected in \Cref{rem:reservoir}). 
The deterministic starter satisfies the same $t+D$ bound, since its gap is exactly $D$. 

\begin{lemma}
    \label{lem:ladder} The deterministic starter requests at $t_0$ and is chosen so that $t_0+D<\min S_1$; every other page's first-request atoms lie in $S_1$ or later. Let $t$ be either $t_0$ (for the initial eviction) or a request time in $S_j$ (for the cascade in block $j$). The resident pages' median next-request epochs are ordered as follows: 
    \begin{center}
    \begin{tabular}{@{}ll@{}}
    \toprule
    resident page & median next-request epoch \\
    \midrule
    spent page or deterministic starter & at most $t+D$ \\
    token whose next atom is a filter & in $G_j$ \\
    token whose next atom is a checkpoint & in $T_j$ \\
    anchor in $P_r$ with $r>j$, including $P_{H+1}$ & in $F_j$ \\
    live anchor in $P_j$ & in $L_j$ \\
    \bottomrule
    \end{tabular}
    \end{center}
    The first row covers the starter and every spent page; the remaining rows classify the live pages by their next unresolved atom and, together, cover all of them. These locations are strictly increasing. 
\end{lemma}
\begin{proof}
    For $t$ in $S_j$, the table follows from \Cref{lem:program,lem:locate}. A spent page has residual lower median at most $D$, by the spent-page convention stated before the lemma, hence median epoch at most $t+D$; the deterministic starter has the same bound. Every other page has the median epoch specified by its next atom, except a live $P_j$ anchor, whose median is its late atom. The empty gap after $S_j$ has length $2D$, so $t+D<\min G_j$ for every $t\in S_j$. Equation~\eqref{eqn:windows} then gives the strict ordering. The same argument applies at $t_0$, because $t_0+D<\min S_1<\min G_1$. 
\end{proof}

\begin{corollary}
    \label{cor:evict}
    Assume the anchor pool named below has a live survivor,\footnote{This is automatic in the idealized experiment and holds in the finite instance on the good event of \Cref{lem:finite}.} then
    \begin{enumerate}[label=(\alph*), itemsep=1pt, topsep=2pt]
        \item At the starter fault at $t_0$ and at every fault of $\mcA_{\medalg}$ during $S_j$, the median policy evicts the smallest-index live anchor of $P_j$ (with $j=1$ at the starter fault).
        \item If the selected anchor, say $p_m$, takes its late branch, then at its request time $t_\ell:=\ell_{p_m}\in L_j$
        every anchor of $P_j$ is spent, and the median policy evicts the smallest-index surviving anchor of $P_{j+1}$, or of $P_{H+1}$ when $j=H$.
    \end{enumerate}
\end{corollary}
\begin{proof}
    (a) At any such $t$, \Cref{lem:ladder} places the live anchors of $P_j$ strictly above every other resident; among them, \eqref{eqn:ranks} gives the smallest-index live anchor the largest late atom, so it is the unique maximum-median page and is evicted.

    (b) We compare the medians at the {\em late fault} $t_\ell$ in $L_j$ directly, via the following two steps. \emph{First, the pool $P_j$ is exhausted at $t_\ell$.} Any anchor that failed one of $F_1,\ldots,F_{j-1}$ was already spent before block $j$. Among the anchors that passed those filters, each smaller-index anchor than $p_m$ was selected earlier and requested early in $S_j$. Each larger-index survivor either requests early in $S_j$ or has a late atom $\ell_{p_i}<\ell_{p_m}$. Thus every page of $P_j$ is spent by time $t_\ell$. 

    \emph{Second, the eviction at $t_\ell$.} Assume first that $j<H$. At this time, each surviving token has its next atom either at a filter in $G_{j+1}$ or at a checkpoint in $T_{j+1}$. Each anchor in a pool later than $P_{j+1}$ has its next filter in $F_{j+1}$, while a surviving $P_{j+1}$ anchor has declined its last filter and has median in $L_{j+1}$. Because the empty gap from $L_j$ to $S_{j+1}$ has length $2D$, every spent page and the deterministic starter have median epoch before $S_{j+1}$, hence before all four displayed candidate groups. The next-block ordering is therefore $t_\ell+D \prec G_{j+1}\prec T_{j+1}\prec F_{j+1}\prec L_{j+1}$, and the smallest-index surviving anchor of $P_{j+1}$ is the unique maximum, hence is evicted. For $j=H$: every surviving token has passed all its checkpoints by $t_\ell\in L_H$, so its median is its post-horizon tail, and every surviving $P_{H+1}$ anchor has declined all its filters, so its median is its post-horizon atom. Our construction places every $P_{H+1}$ atom after every token tail and gives smaller-index boundary anchors later atoms, while spent pages and the starter again have median epoch at most $t_\ell+D$, which precedes all tails. Hence the smallest-index surviving boundary anchor is the unique maximum and is evicted.
\end{proof}

\subsection{The Median and Comparison Policies}
\label{sec:med-dynamics}

The ``inexhaustible pools'' used below are an abstract \emph{idealized} experiment, defined precisely in \Cref{rem:reservoir}, not a paging instance with infinitely many simultaneously running renewal processes. For each pool, reveal an ordered sequence of independent candidates only as the dynamics needs them; a revealed anchor has the conditional early/late law $(a,b)$ after surviving its filters, and a revealed token candidate has conditional checkpoint probability $1/2$. The finite construction is the actual paging instance. On the good event in \Cref{lem:finite}, its first sufficiently many surviving candidates have exactly the same revealed outcomes as this reservoir experiment. This interpretation avoids any local-finiteness issue from literally merging infinitely many renewal processes.

\begin{definition}[Idealized experiment]\label{rem:reservoir}
For definiteness, the idealized experiment makes the following assumptions; the finite
instance matches each of them on the good event of \Cref{lem:finite}, which is the
content of the coupling in its proof.
\begin{enumerate}[label=(\roman*), itemsep=1pt, topsep=2pt]
    \item It runs on the timetable of \Cref{sec:timing}: the same windows, residue
    classes, atom placements, and deterministic starter.
    \item Every post-first gap equals $D$, so a spent page requests deterministically
    every $D$ time units. This is the convention invoked for spent pages before
    \Cref{lem:ladder}; in the finite instance it is enforced, before the horizon, by
    excluding long post-first gaps; this is precisely the bad event $\mathcal E_3$ in the proof of \Cref{lem:finite}.
    \item Pools are inexhaustible: whenever the dynamics needs the next live anchor of
    $P_j$, a filter-surviving replacement token of $Q_{j+1}$, or a surviving boundary
    anchor of $P_{H+1}$, one exists. This is exactly the survivor hypothesis of
    \Cref{cor:evict}.
    \item The experiment ends at the final late fault of block $H$; every post-horizon
    atom lies strictly after all times considered and is never requested.
    \item All revealed candidate outcomes are independent, with the conditional laws
    stated above.
\end{enumerate}
\end{definition}

We compare the median policy $\mcA_{\medalg}$ with an online comparison policy $\mcA_{\cmpalg}$, whose evictions are defined here. In the finite construction, every anchor and token starts in the cache, while the deterministic starter page $z$ starts outside, so the cache contains all but one page. There are $N=3L(2^{H+1}-1)+1$ pages in total, and the cache size is $k=N-1$. By construction, $z$ requests first at $t_0$ and then every $D$ time units. The first request is therefore to $z$, and these two policies fault. By \Cref{cor:evict}(a), $\mcA_{\medalg}$ evicts the smallest-index live anchor $p\in P_1$. The comparison policy $\mcA_{\cmpalg}$ instead evicts a fixed token $q\in Q_1$. If $C$ is the common part of the two caches after this fault, then
\begin{align}
    \label{eqn:gap-state}
    \mcA_{\cmpalg}=C\cup\{p\},
    \qquad
    \mcA_{\medalg}=C\cup\{q\}.
\end{align}
We call $p$ the \emph{missing anchor} of the median policy and $q$ the \emph{active token} of the comparison policy.   \Cref{fig:pair-state} previews the resulting one-pair cache invariant and separates the two possible outcomes at the checkpoint $T_j$ ({\em whether the active token requests there or it does not}) before the formal proof. 

\paragraph{The comparison policy $\mcA_\cmpalg$.} The policy is fully specified by three clauses: at the starter fault it evicts a fixed token of $Q_1$ (which is checkpoint-ready in block $1$, since $Q_1$ tokens have no filters); at each of its checkpoint faults (that is, whenever its active token requests at a checkpoint $T_j$), it installs the requested token and evicts the first token of $Q_{j+1}$ that has survived all filters through $G_j$ (survival is determined by the observed request history, so the rule is online, and the evicted token is checkpoint-ready in block $j+1$); and if no such token exists (which happens only on the bad event excluded in \Cref{lem:finite}) it follows a fixed fallback rule, say, evict the smallest-index token currently in its cache, or the smallest-index page if no token remains; the choice is immaterial, since \Cref{lem:finite} bounds costs on the bad event crudely. It never evicts any other page, and we will see in \Cref{lem:invariant} that it never faults except at $t_0$ and at checkpoints.

This definition is also where the separation $G_j\prec T_j$ inside each block, promised in \Cref{sec:anchor-token}, earns its keep, in two ways. First, at its checkpoint fault in $T_j$ the policy must name a token of $Q_{j+1}$ that has survived all filters through $G_j$; because every $G_j$ atom precedes every $T_j$ atom, that survival is fully resolved by the observed history before the fault, keeping the rule online. Second, the replacement token it evicts has no atom before $T_{j+1}$, so it cannot request and force an unintended second comparison fault later in the same block. If the token filters shared the checkpoint window, both properties could fail.

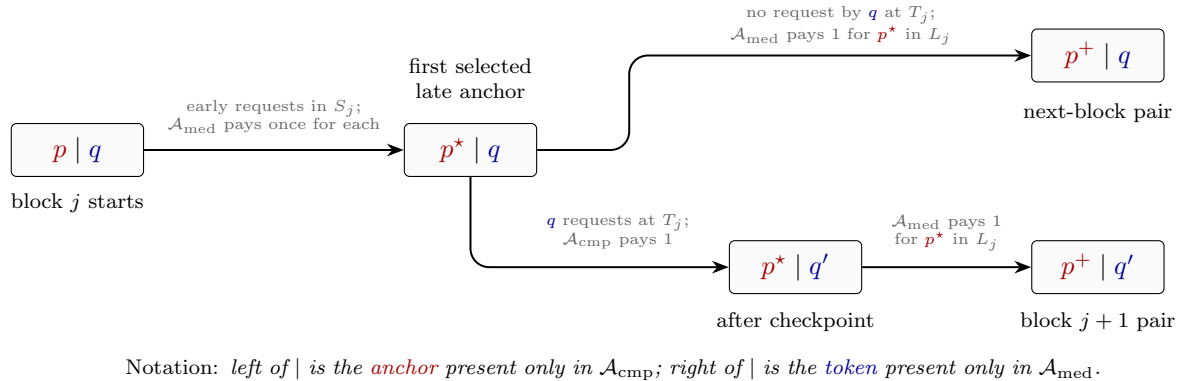
\begin{figure}[H]
\centering
\begin{tikzpicture}[
  font=\small,
  >=Stealth,
  state/.style={draw, rounded corners=2pt, align=center,
                minimum width=1.75cm, minimum height=0.7cm, fill=gray!4},
  lab/.style={align=center, font=\tiny, text=black!60, fill=white, inner sep=1pt},
  scap/.style={align=center, font=\scriptsize}
]
  \node[state] (start) at (0,0) {${\color{red!65!black}p}\mid{\color{blue!60!black}q}$};
  \node[scap] at (0,-0.7) {block $j$ starts};

  \node[state] (late) at (5.2,0) {${\color{red!65!black}p^\star}\mid{\color{blue!60!black}q}$};
  \node[scap, anchor=south] at (5.2,0.55)
    {first selected\\late anchor};

  \node[state] (upper) at (13.5,1.25) {${\color{red!65!black}p^+}\mid{\color{blue!60!black}q}$};
  \node[scap] at (13.5,0.5) {next-block pair};

  \node[state] (lower) at (9.5,-1.55) {${\color{red!65!black}p^\star}\mid{\color{blue!60!black}q'}$};
  \node[scap] at (9.5,-2.27) {after checkpoint};

  \node[state] (lowerright) at (13.5,-1.55) {${\color{red!65!black}p^+}\mid{\color{blue!60!black}q'}$};
  \node[scap] at (13.5,-2.27) {block $j+1$ pair};

  \draw[->, line width=0.75pt] (start) --
      node[lab, above=5pt, text width=2.9cm]
      {early requests in $S_j$;\\$\mcA_{\medalg}$ pays once for each}
      (late);

  \draw[->, line width=0.75pt, rounded corners=7pt]
      (late.east) -- (7.3,0) -- (7.3,1.25) -- (upper.west);
  \node[lab, anchor=south, text width=4.4cm] at (10.1,1.35)
      {no request by ${\color{blue!60!black}q}$ at $T_j$;\quad $\mcA_\medalg$ pays $1$ for ${\color{red!65!black}p^\star}$ in $L_j$};

  \draw[->, line width=0.75pt, rounded corners=7pt]
      (late.south) -- (5.2,-1.55) -- (lower.west);
  \node[lab, anchor=south, text width=2.7cm] at (7.15,-1.35)
      {${\color{blue!60!black}q}$ requests at $T_j$;\\$\mcA_\cmpalg$ pays $1$};

  \draw[->, line width=0.75pt] (lower) --
      node[lab, above=5pt, text width=2.1cm]
      {$\mcA_\medalg$ pays $1$\\for ${\color{red!65!black}p^\star}$ in $L_j$}
      (lowerright);

  \node[anchor=west, align=left, font=\scriptsize, text width=14.6cm] at (0.5,-2.9)
      {Notation: {\em left of $\mid$ is the {\color{red!65!black}anchor} present only in $\mcA_{\cmpalg}$; right of $\mid$ is the {\color{blue!60!black}token} present only in $\mcA_{\medalg}$.}};
\end{tikzpicture}
\caption{The one-pair invariant as a state machine. Early anchor requests change only the missing anchor. A token checkpoint changes only the active token. The selected anchor's late return advances the anchor discrepancy to the next block. Hence the caches always differ by one anchor/token pair, while $\mcA_\medalg$ pays once per selected anchor and $\mcA_\cmpalg$ pays at most once per block.}
\label{fig:pair-state}
\end{figure}

\begin{lemma}[One-pair invariant]
    \label{lem:invariant}
    Suppose block $j$ begins with caches of the form \eqref{eqn:gap-state}, where $p$ is the smallest-index live anchor of $P_j$ and the next possible request of $q$ is the checkpoint $T_j$ (i.e., $q$ is checkpoint-ready in block $j$). Throughout the block, the two caches differ by exactly one missing-anchor/active-token pair. The median policy $\mcA_\medalg$ faults once for each anchor it selects, whereas the comparison policy $\mcA_\cmpalg$ faults in the block if and only if its active token $q$ requests at $T_j$, in which case it faults once. 
    
    Moreover, if $j<H$, then block $j+1$ again begins with caches of the form \eqref{eqn:gap-state}: the missing anchor is the smallest-index live anchor of $P_{j+1}$, and the active token is checkpoint-ready in block $j+1$.
\end{lemma}
\begin{proof}
    First suppose the missing anchor $p$ requests early in $S_j$. It is present only in $\mcA_\cmpalg$'s cache, so only $\mcA_\medalg$ faults. By \Cref{cor:evict}(a), $\mcA_\medalg$ evicts the next smallest-index live anchor $p'$ of $P_j$. The caches again have the form \eqref{eqn:gap-state}, with discrepancy $p'$ versus the same token $q$. This argument repeats for every selected anchor taking the early branch.

    Eventually the selected anchor $p^*$ takes its late branch. 
    Consider the checkpoint first: if $q$ requests at $T_j$, only $\mcA_{\cmpalg}$ faults; it installs $q$ and evicts a surviving token $q'\in Q_{j+1}$, so the active token becomes $q'$ while the missing anchor is still $p^*$. Whether or not this happened, exactly one median fault follows: at the request to $p^*$ in $L_j$, only $\mcA_\medalg$ faults, and by \Cref{cor:evict}(b) it evicts an anchor $p^{+}$ of $P_{j+1}$ (or $P_{H+1}$ in the last block). The discrepancy is then $p^{+}$ versus the current active token, $q$ or $q'$. Every other request is to a page held by both policies and changes neither cache discrepancy nor cost. \Cref{fig:pair-state} depicts these transitions. 

    Finally, the closing configuration restores the hypothesis for block $j+1$ when $j<H$: the missing anchor $p^{+}$ is the smallest-index live anchor of $P_{j+1}$ by \Cref{cor:evict}(b), and the active token is checkpoint-ready in block $j+1$---it is either $q$, which declined $T_j$ and therefore has its next unresolved atom in $T_{j+1}$, or the replacement $q'\in Q_{j+1}$, which survived all filters through $G_j$ and has its first checkpoint atom in $T_{j+1}$.
\end{proof}

\begin{proposition}
    \label{prop:ideal-costs}
    In the abstract idealized experiment, the expected costs through the end of block $H$ are
    \begin{align}
        \label{eqn:ideal-costs}
        \E[\cost(\mcA_{\medalg})] = 1 + \frac Hb,
        \qquad
        \E[\cost(\mcA_{\cmpalg})] = 1 + \frac H2 .
    \end{align}
\end{proposition}
\begin{proof}
    Here an anchor counts as \emph{selected in block $j$} if it is the missing anchor at some moment of block $j$; the anchor evicted at the previous block's late fault (or at the starter fault when $j=1$) is the first selection of block $j$, and each selected anchor's single first request faults the median policy during the block, by \Cref{lem:invariant}.
    Conditioned on reaching its terminal pair, every selected anchor independently takes the early branch with probability $a$ and the late branch with probability $b$. Therefore the number $A_j$ of anchors selected in block $j$ satisfies
    \[
        \Pr[A_j=m]=a^{m-1}b,
        \qquad m=1,2,\ldots,
        \qquad
        \E[A_j]=\frac{1}{b}.
    \]
    Conditional on being active at the beginning of block $j$, the token $q$ places exactly half of its remaining probability at $T_j$. Thus its checkpoint-fault indicator $B_j$ has conditional expectation $1/2$, and hence $\E[B_j]=1/2$. Independence among the $B_j$ is not needed. By \Cref{lem:invariant}, the two costs are the common starter fault plus $\sum_{j=1}^H A_j$ and $\sum_{j=1}^H B_j$, respectively. Linearity of expectation proves \eqref{eqn:ideal-costs}.
\end{proof}

\subsection{Finite Instances and a Request-Count Horizon}
\label{sec:finite}

Recall from \Cref{sec:timing} the integer $L$, the pool sizes \eqref{eqn:pools}, namely $|P_r|=2L2^{r-1}$ and $|Q_r|=L2^{r-1}$, and the page total $N=3L(2^{H+1}-1)+1$.
A member of $P_r$ survives its first $r-1$ filters with probability $2^{-(r-1)}$, so the expected number of terminal-pair survivors is exactly $|P_r|\cdot2^{-(r-1)}=2L2^{r-1}\cdot2^{-(r-1)}=2L$. A member of $Q_r$ survives to its first checkpoint $T_r$, or to its post-horizon tail when $r=H+1$, with probability $2^{-(r-1)}$, so the expected number of such survivors is exactly $|Q_r|\cdot2^{-(r-1)}=L$. For fixed $H$, the number of pages is polynomial in $L$.

Throughout, a \emph{request horizon} $T$ is a number of requests: the run is truncated after the $T$th request, and $\cost(\pi,T)$ counts faults among those requests, as in \eqref{eqn:objective}. 

\begin{lemma}
    \label{lem:finite}
    For fixed $H$ and $\gamma$, the atoms, programming parameters, and a finite request horizon $T_L$ can be chosen so that
    \begin{align}
        \E[\cost(\mcA_{\medalg},T_L)] &= 1 + \frac Hb + o_L(1), \\
        \E[\cost(\mcA_{\cmpalg},T_L)] &= 1 + \frac H2 + o_L(1). \label{eqn:finite-costs}
    \end{align}
\end{lemma}
\begin{proof}
    \emph{Step 1: the bad events.} Let $\mathcal E_1,\mathcal E_2,\mathcal E_3$ be the following bad events.
    \begin{enumerate}[label=(\roman*), leftmargin=2em, itemsep=2pt]
        \item[$\mathcal E_1$:] Some anchor pool has fewer than $L$ survivors after its filters, or some token pool has no survivor reaching its first checkpoint (or, for $Q_{H+1}$, its post-horizon tail).
        \item[$\mathcal E_2$:] In the idealized experiment, $A_j>L$ for some block $j$; that is, the block would require more than $L$ selected anchors.
        \item[$\mathcal E_3$:] Before the chosen request horizon, some programmed page draws a post-first gap different from the short gap $D$.
    \end{enumerate}
    
    \smallskip
    \emph{Step 2: bounding $\Pr[\mathcal E_1]$ and $\Pr[\mathcal E_2]$.} The relevant survivor counts in $\mathcal E_1$ are binomial with means $2L$ for anchors and $L$ for tokens. Explicitly, the number of terminal-pair survivors of $P_r$ is $S_r\sim\mathrm{Bin}(2L2^{r-1},2^{-(r-1)})$ with $\E[S_r]=2L$; the multiplicative Chernoff lower tail $\Pr[S_r\le(1-\beta)\E[S_r]]\le e^{-\beta^2\E[S_r]/2}$ with $\beta=1/2$ gives $\Pr[S_r<L]\le e^{-L/4}$. The probability that no token of $Q_r$ reaches its first checkpoint (or, for $Q_{H+1}$, its tail) is $(1-2^{-(r-1)})^{L2^{r-1}}\le e^{-L}$, by $1-x\le e^{-x}$. A union bound over the $2(H+1)$ pools gives $\Pr[\mathcal E_1]\le 2(H+1)e^{-L/4}$. 
    In each block, selecting more than $L$ anchors requires the first $L$ selected anchors all to take the early branch, so $\Pr[\mathcal E_2]\le Ha^L$.

    \smallskip
    \emph{Step 3: the request-count horizon.} We next choose a horizon measured by requests. All intended first-request atoms through $L_H$ can be placed at polynomially bounded times; let $\tau_L$ be the integer {\em time} just after $L_H$. 
    Every programmed page has at most one first request and every later gap is at least $D$; the starter requests exactly every $D$ time units. Hence $B_L:=N (1+\lceil{\tau_L}/{D}\rceil)$ is an upper bound on the total number of requests by time $\tau_L$. Place every post-horizon atom after $\tau_L+2(B_L+1)D$ and set $T_L=B_L+1$. Because at most $B_L$ requests occur by $\tau_L$, the $T_L$th request occurs after all intended atoms through $L_H$. Conversely, the deterministic starter alone makes at least $B_L+1$ requests during any interval of length $(B_L+1)D$, so the $T_L$th request occurs no later than $\tau_L+(B_L+1)D$, strictly before every post-horizon atom.

    \smallskip
    \emph{Step 4: bounding $\Pr[\mathcal E_3]$.} Use the same programming probability $\delta_L=o(T_L^{-2})$ for every programmed page. The deterministic starter contributes no bad gap, and at most $T_L$ post-first gaps of programmed pages are sampled before the horizon, so a union bound gives $\Pr[\mathcal E_3]\le T_L\delta_L=o(1/T_L)$.

    \smallskip
    \emph{Step 5: coupling to the idealized experiment.}  Couple the finite instance to it by ordering each finite pool by rank and matching its successive filter survivors to the successive candidates revealed from the corresponding idealized pool. Match the early/late terminal branch of each anchor and the complete first-checkpoint-or-tail outcome of each token. On the event $\mathcal G_L :=(\mathcal E_1\cup\mathcal E_2\cup\mathcal E_3)^c$,
    every anchor and replacement token selected by either policy belongs to the finite pools, and every post-first gap seen before the horizon is the short gap $D$. Hence the finite-instance and idealized costs agree on $\mathcal G_L$. In detail: on $\mathcal G_L$ the matched outcomes and uniformly short gaps make the two request sequences identical through the final late fault, so, by induction over requests, both policies evict and fault identically in the two experiments. After that fault neither experiment faults again before the horizon: with cache size $k=N-1$, each cache misses exactly one page ($\mcA_\medalg$'s cache misses the boundary anchor, $\mcA_\cmpalg$'s misses the active token) and neither of these two pages has a remaining atom before the horizon, while every other page is in both caches. Finally, the late fault precedes $\tau_L$, hence the $T_L$th request (Step 3), so every idealized-experiment fault is counted within the horizon.
    
    \smallskip
    \emph{Step 6: from the coupling to the cost estimates.} Let $C_L^\pi$ and $C_\infty^\pi$ denote the finite-instance and idealized costs for $\pi\in\{\mcA_{\medalg},\mcA_{\cmpalg}\}$. We always have $C_L^\pi\le T_L$, since the run contains $T_L$ requests and each contributes at most one fault. Moreover, for fixed $H$ and $\gamma$, $C_\infty^{\mcA_{\cmpalg}}\le H+1$, while  $C_\infty^{\mcA_{\medalg}}=1+\sum_{j=1}^H A_j$ (the block selection counts of \Cref{prop:ideal-costs}) has second moment:
        \begin{align*}
        \E\bigl[(C_\infty^{\mcA_{\medalg}})^2\bigr]
        &= \operatorname{Var}\bigl(C_\infty^{\mcA_{\medalg}}\bigr)+\bigl(\E\bigl[C_\infty^{\mcA_{\medalg}}\bigr]\bigr)^2
            && \\
        &= \operatorname{Var}\Bigl(\textstyle\sum_{j=1}^H A_j\Bigr)+(1+H/b)^2
            && \rhd~\text{$\E[A_j]=1/b$}\\
        &= \textstyle\sum_{j=1}^H \operatorname{Var}(A_j)+(1+H/b)^2
            && \rhd~\text{$\{A_j\}_{1\ldots H}$ are independent}\\
        &= Ha/b^2+(1+H/b)^2.
            && \rhd~\text{$\operatorname{Var}(A_j)=a/b^2$ as computed below}
    \end{align*}
     Therefore, writing $\Delta:=C_L^\pi-C_\infty^\pi$, Cauchy--Schwarz gives
    \begin{align*}
        \left|\E[C_L^\pi]-\E[C_\infty^\pi]\right|
        &= \bigl|\E[\Delta\,\Ind[\mathcal G_L]]+\E[\Delta\,\Ind[\mathcal G_L^c]]\bigr|
             &&\rhd~\text{$\Ind[\mathcal G_L]{+}\Ind[\mathcal G_L^c]=1$}\\
        &= \bigl|\E[\Delta\,\Ind[\mathcal G_L^c]]\bigr|
             &&\rhd~\text{$\Delta=0$ on $\mathcal G_L$ (Step 5)}\\
        &\le \E[C_L^\pi\Ind[\mathcal G_L^c]]
             +\E[C_\infty^\pi\Ind[\mathcal G_L^c]]
             &&\rhd~\text{triangle inequality; nonnegativity}\\
        &\le T_L\Pr[\mathcal G_L^c]
             +\sqrt{\E[(C_\infty^\pi)^2]\Pr[\mathcal G_L^c]}
             &&\rhd~\text{$C_L^\pi\le T_L$; Cauchy--Schwarz}\\
        &=o(1).
             &&\rhd~\text{Steps 2 and 4; $T_L=\operatorname{poly}(L)$}
    \end{align*}
    To elaborate the final step, the two summands vanish for different reasons: the first races the polynomial $T_L$ against small probabilities, while in the second the square root of the second moment is a constant not depending on $L$, so plain convergence $\Pr[\mathcal G_L^c]\to0$ suffices. For the first summand,
    \begin{align*}
        T_L\Pr[\mathcal G_L^c]
        &\le T_L\bigl(\Pr[\mathcal E_1]+\Pr[\mathcal E_2]+\Pr[\mathcal E_3]\bigr)
            && \rhd~\text{union bound}\\
        &\le T_L\bigl(2(H{+}1)e^{-L/4}+Ha^L\bigr)+T_L^2\delta_L
            && \rhd~\text{Steps 2 and 4}\\
        &= o(1),
            && 
    \end{align*}
    where the last step holds since $T_L=\operatorname{poly}(L)$ for fixed $H$ (Step 3), and  $\delta_L=o(T_L^{-2})$. For the second summand,
    \begin{align*}
        \sqrt{\E[(C_\infty^\pi)^2]\Pr[\mathcal G_L^c]}
        &= \sqrt{\E[(C_\infty^\pi)^2]}\cdot\sqrt{\Pr[\mathcal G_L^c]}
            && \rhd~\text{split the square root}\\
        &\le M_{H,\gamma}\sqrt{\Pr[\mathcal G_L^c]}
            && \rhd~\text{the idealized experiment does not involve $L$}\\
        &= o(1),
            && \rhd~\text{$\Pr[\mathcal G_L^c]\to0$ by the three bounds above}
    \end{align*}
    where $M_{H,\gamma}^2:=\max\bigl\{(H{+}1)^2,\,(1+H/b)^2+Ha/b^2\bigr\}$ bounds $\E[(C_\infty^\pi)^2]$ for both policies and is a constant for fixed $H$ and $\gamma$. Combining this coupling estimate with \Cref{prop:ideal-costs} proves both equalities in the lemma.
\end{proof}

\subsection{Proof of \texorpdfstring{\Cref{thm:med-lb}}{Theorem 2}}
\label{sec:med-wrapup}

\thmmedlb*

\begin{proof}
    Since $\opt$ minimizes \eqref{eqn:objective} over online policies and $\mcA_{\medalg}$ is online, $\E[\cost(\opt,T)]\le\E[\cost(\mcA_{\medalg},T)]$ on every instance and horizon, so the competitive ratio of $\mcA_{\medalg}$ is at least $1$. Fix any $c\in[1,4)$; by \Cref{obs:online-witness} it suffices to exhibit an instance, a horizon, and an online policy $\mcB$ with $\E[\cost(\mcA_{\medalg},T)]>c\,\E[\cost(\mcB,T)]$. Our witness is $\mcB=\mcA_{\cmpalg}$, which is online.

    \smallskip
    \emph{Choice of $\gamma$.} Since $c<4$, the number $(4-c)/(2c)$ is positive; fix $\gamma\in(0, \min\{\tfrac12, \tfrac{4-c}{2c}\})$, so that $\frac2b=\frac4{1+2\gamma}>c$;
    equivalently, $\gamma<(4-c)/(2c)$, using $b=\tfrac12+\gamma$.

    \smallskip
    \emph{Choice of $H$.} With $\gamma$ fixed, the ratio $(1+H/b)/(1+H/2)$ is increasing in $H$ and
    \[
        \lim_{H\to\infty}\frac{1+H/b}{1+H/2}=\frac{1/b}{1/2}=\frac2b>c,
    \]
    so there is a fixed $H\ge1$ with $(1+H/b)/(1+H/2)>c$.

    \smallskip
    \emph{Choice of $L$.} With $H$ and $\gamma$ held fixed, write the conclusion of \Cref{lem:finite} as
    \[
        \E[\cost(\mcA_{\medalg},T_L)]=1+\tfrac Hb+\varepsilon_L,
        \qquad
        \E[\cost(\mcA_{\cmpalg},T_L)]=1+\tfrac H2+\varepsilon'_L,
    \]
    where $\varepsilon_L$ and $\varepsilon'_L$ are real sequences tending to $0$ as $L\to\infty$; this is the meaning of $o_L(1)$ in that lemma. Set $\eta:=\bigl(1+H/b\bigr)-c\bigl(1+H/2\bigr)$. The choice of $H$ gives $(1+H/b)/(1+H/2)>c$, and $1+H/2>0$, so $\eta>0$. Since both error sequences vanish, we may fix $L$ large enough that $|\varepsilon_L|\le\eta/4$ and $|\varepsilon'_L|\le\eta/(4c)$. Then
    \[
        \E[\cost(\mcA_{\medalg},T_L)]-c\,\E[\cost(\mcA_{\cmpalg},T_L)]
        =\eta+\varepsilon_L-c\,\varepsilon'_L
        \ge\eta-\tfrac\eta4-\tfrac\eta4
        =\tfrac\eta2>0 .
    \]

    For this instance and horizon, \Cref{obs:online-witness} applied with the online witness $\mcA_{\cmpalg}$ shows that $\mcA_{\medalg}$ is not $c$-competitive. As $c\in[1,4)$ was arbitrary, the competitive ratio of the median algorithm is at least $4$.
\end{proof}

    The programming lemma realizes every prescribed first-request law by using an i.i.d. finite-support gap distribution together with a known deterministic initial age. These initial ages are essential to the present construction.
    
    This initialization convention is also covered by the matching upper bounds. Lund et al.~\cite{lund1999paging} state their $5$-competitive theorem for every initial configuration, explicitly including each page's elapsed time since its last request. Pabbaraju and Vakilian's improved charging lemma~\cite{pabbaraju2025new} uses only the same per-fault precedence guarantee; conditional on the full observed history (equivalently, the current age vector) independent renewal processes have independent residual times, so the median rule satisfies that guarantee with $c=4$. Hence their $4$-competitive upper bound also holds for arbitrary known deterministic initial ages.

\section*{AI Disclosure} 
The results were obtained through iterative and interactive engagement between the author and AI models, including GPT-5.5 Thinking, GPT-5.6 Sol, Opus 4.8 and 5, and Fable 5. AI models were also used throughout the writing process, after which the manuscript was substantially revised and rewritten by the author. All figures were generated using Opus 4.8 and Opus 5 and were subsequently modified and refined by the author.

The author takes full responsibility for the correctness of the results and has extensively verified, modified, and improved the proofs and exposition throughout the process, building on AI-generated proof ideas, analyses, and drafts.

\bibliographystyle{plain}
\bibliography{ref}

@article{belady1966study,
  title={A study of replacement algorithms for a virtual-storage computer},
  author={Belady, Laszlo A.},
  journal={IBM Systems journal},
  volume={5},
  number={2},
  pages={78--101},
  year={1966},
  publisher={IBM}
}

@article{sleator1985amortized,
  title={Amortized efficiency of list update and paging rules},
  author={Sleator, Daniel D and Tarjan, Robert E},
  journal={Communications of the ACM},
  volume={28},
  number={2},
  pages={202--208},
  year={1985},
  publisher={ACM New York, NY, USA}
}

@inproceedings{karlin1992markov,
  title={Markov paging},
  author={Karlin, AR and Phillips, SJ and Raghavan, P},
  booktitle={Proceedings of the 33rd Annual Symposium on Foundations of Computer Science},
  pages={208--217},
  year={1992}
}

@article{lund1999paging,
  title={Paging against a Distribution and {IP} Networking},
  author={Lund, Carsten and Phillips, Steven and Reingold, Nick},
  journal={Journal of Computer and System Sciences},
  volume={58},
  number={1},
  pages={222--231},
  year={1999},
  publisher={Elsevier}
}

@book{borodin2005online,
  title={Online computation and competitive analysis},
  author={Borodin, Allan and El-Yaniv, Ran},
  year={2005},
  publisher={cambridge university press}
}

@misc{roughgardennotes,
  author = {Tim Roughgarden},
  title = {{CS 369N Beyond Worst-Case Analysis, Lecture Notes 2: Models of Data in Online Paging}},
  howpublished = "\url{https://timroughgarden.org/f11/l2.pdf}",
  year = {2010}
}

@inproceedings{fiat1995randomized,
  title={Randomized and multipointer paging with locality of reference},
  author={Fiat, Amos and Karlin, Anna R},
  booktitle={Proceedings of the twenty-seventh annual ACM symposium on Theory of computing},
  pages={626--634},
  year={1995}
}

@inproceedings{irani1992strongly,
  title={Strongly competitive algorithms for paging with locality of reference},
  author={Irani, Sandy and Karlin, Anna R and Phillips, Steven},
  booktitle={Proceedings of the third annual ACM-SIAM symposium on Discrete algorithms},
  pages={228--236},
  year={1992}
}

@inproceedings{borodin1991competitive,
  title={Competitive paging with locality of reference},
  author={Borodin, Allan and Raghavan, Prabhakar and Irani, Sandy and Schieber, Baruch},
  booktitle={Proceedings of the twenty-third annual ACM symposium on Theory of computing},
  pages={249--259},
  year={1991}
}

@article{koutsoupias2000beyond,
  title={Beyond competitive analysis},
  author={Koutsoupias, Elias and Papadimitriou, Christos H},
  journal={SIAM Journal on Computing},
  volume={30},
  number={1},
  pages={300--317},
  year={2000},
  publisher={SIAM}
}

@inproceedings{becchetti2004modeling,
  title={Modeling locality: A probabilistic analysis of LRU and FWF},
  author={Becchetti, Luca},
  booktitle={Algorithms--ESA 2004: 12th Annual European Symposium, Bergen, Norway, September 14-17, 2004. Proceedings 12},
  pages={98--109},
  year={2004},
  organization={Springer}
}

@article{torng1998unified,
  title={A unified analysis of paging and caching},
  author={Torng, Eric},
  journal={Algorithmica},
  volume={20},
  pages={175--200},
  year={1998},
  publisher={Springer}
}

@article{young2000line,
  title={On-line paging against adversarially biased random inputs},
  author={Young, Neal E},
  journal={Journal of Algorithms},
  volume={37},
  number={1},
  pages={218--235},
  year={2000},
  publisher={Elsevier}
}

@article{denning1968working,
  title={The working set model for program behavior},
  author={Denning, Peter J},
  journal={Communications of the ACM},
  volume={11},
  number={5},
  pages={323--333},
  year={1968},
  publisher={ACM New York, NY, USA}
}

@article{irani2005competitive,
  title={Competitive analysis of paging},
  author={Irani, Sandy},
  journal={Online Algorithms: The State of the Art},
  pages={52--73},
  year={2005},
  publisher={Springer}
}

@article{dorrigiv2005survey,
  title={A survey of performance measures for on-line algorithms},
  author={Dorrigiv, Reza and L{\'o}pez-Ortiz, Alejandro},
  journal={SIGACT News},
  volume={36},
  number={3},
  pages={67--81},
  year={2005}
}

@inproceedings{angelopoulos2009paging,
  title={Paging and list update under bijective analysis},
  author={Angelopoulos, Spyros and Schweitzer, Pascal},
  booktitle={Proceedings of the twentieth annual ACM-SIAM symposium on Discrete algorithms},
  pages={1136--1145},
  year={2009},
  organization={SIAM}
}

@inproceedings{young1998bounding,
  title={Bounding the diffuse adversary},
  author={Young, Neal E},
  booktitle={Proceedings of the ninth annual ACM-SIAM symposium on Discrete algorithms},
  pages={420--425},
  year={1998}
}

@article{hao2018learning,
  title={On learning markov chains},
  author={Hao, Yi and Orlitsky, Alon and Pichapati, Venkatadheeraj},
  journal={Advances in Neural Information Processing Systems},
  volume={31},
  year={2018}
}

@article{pabbaraju2025new,
  title={New and Improved Bounds for Markov Paging},
  author={Pabbaraju, Chirag and Vakilian, Ali},
  journal={arXiv preprint arXiv:2502.05511},
  year={2025}
}

@article{fiat1991competitive,
  title={Competitive paging algorithms},
  author={Fiat, Amos and Karp, Richard M and Luby, Michael and McGeoch, Lyle A and Sleator, Daniel D and Young, Neal E},
  journal={Journal of Algorithms},
  volume={12},
  number={4},
  pages={685--699},
  year={1991},
  publisher={Elsevier}
}

@article{wolfer2021statistical,
  title={Statistical estimation of ergodic Markov chain kernel over discrete state space},
  author={Wolfer, Geoffrey and Kontorovich, Aryeh},
  journal={Bernoulli},
  volume={27},
  number={1},
  pages={532--553},
  year={2021},
  publisher={International Statistical Institute}
}

\clearpage
\appendix
\section{Stationary Calculation for the Four-Page Construction}
\label{app:block-rates}

We prove \Cref{prop:block-rates}. After serving a request, let $s$ be the current page (i.e.; the current state of the Markov chain which is the page just requested) and let $m\ne s$ be the unique missing page. Classify the state by whether $s$ is an entry page ($E$) or an exit page ($X$), and by whether $m$ equals $\nxt(s)$, $\lst(s)$, or $\thr(s)$. This gives six symmetry classes: $E_n,E_\ell,E_t,X_n,X_\ell,X_t$.

Consider the two-parameter family of policies that evict $\nxt(r)$ with probability $\beta_E$ at an entry-page miss and with probability $\beta_X$ at an exit-page miss, and otherwise evict $\lst(r)$. The witness $\mcA_\blkalg$ has $(\beta_E,\beta_X)=(\frac{1}{2},\frac{1}{2(1-\varepsilon)})$, while $\mcA_\latealg$ has $(\beta_E,\beta_X)=(0,0)$. The six-state transition table is
\begin{align}
    E_n&\longrightarrow \beta_X X_n+(1-\beta_X)X_\ell,
        &&\text{with a fault},\nonumber\\
    E_\ell&\longrightarrow X_\ell,
        E_t\longrightarrow X_t,
        &&\text{without a fault},\nonumber\\
    X_n&\longrightarrow(1-\varepsilon)
        [\beta_EE_n+(1-\beta_E)E_\ell]+\varepsilon E_t,
        &&\text{with fault probability }1-\varepsilon,\label{eqn:block-transitions}\\
    X_\ell&\longrightarrow(1-\varepsilon)E_\ell+\varepsilon E_n,
        &&\text{without a fault},\nonumber\\
    X_t&\longrightarrow(1-\varepsilon)E_t
        +\varepsilon[\beta_EE_n+(1-\beta_E)E_\ell],
        &&\text{with fault probability }\varepsilon.\nonumber
\end{align}

For completeness, the lines follow directly from the definitions. In $E_n$, the next request is the missing page $\nxt(s)$, an exit-page miss, and the eviction produces $X_n$ or $X_\ell$. In $E_\ell$ and $E_t$, the next request is a hit and the missing page retains type $\ell$ or $t$. In $X_n$, staying in the block, which has probability $1-\varepsilon$, requests the missing page and causes an entry-page miss; switching blocks, which has probability $\varepsilon$, hits the other block's entry page and leaves a type-$t$ hole. In $X_\ell$, both possible requests are hits, yielding $E_\ell$ or $E_n$. Finally, in $X_t$, staying is a hit of type $E_t$, whereas switching requests the missing page and causes an entry-page miss.

\begin{lemma}
    \label{lem:block-stationary}
    Set $H:=1-\beta_E\beta_X- \varepsilon(1-\beta_X)$.
    When $\beta_X>0$ and $\beta_E<1$, the six-state chain is
    irreducible. Its stationary law, in the order
    $(E_n,E_\ell,E_t,X_n,X_\ell,X_t)$, satisfies, with
    $\alpha:=\pi(E_n)$,
    \begin{align*}
        \pi(X_n)=\pi(E_t)=\pi(X_t)&=\beta_X\alpha,\\
        \pi(E_\ell)&=\frac{\alpha H}{\varepsilon},\\
        \pi(X_\ell)&=(1-\beta_X)\alpha
            +\frac{\alpha H}{\varepsilon},\\
        \alpha&=
        \frac{\varepsilon}
        {2[\varepsilon(1+\beta_X)+H]}.
    \end{align*}
    Its stationary fault rate is
    \begin{align}
        \label{eqn:block-general-rate}
        \rho(\beta_E,\beta_X)
        =\frac{\varepsilon(1+\beta_X)}
        {2[\varepsilon(1+\beta_X)+H]}.
    \end{align}
    For $(\beta_E,\beta_X)=(0,0)$, the chain has a unique closed
    recurrent class $\{E_n,E_\ell,X_\ell\}$ and fault rate
    $\varepsilon/2$.
\end{lemma}

\begin{proof}
    When $\beta_X>0$ and $\beta_E<1$, inspection of
    \eqref{eqn:block-transitions} shows that all six states communicate:
    $E_n$ reaches both $X_n$ and $X_\ell$; $X_n$ reaches
    $E_n,E_\ell,E_t$; and $E_t,X_t$ connect back through a block
    switch. Thus the chain is irreducible.

    Write $(a,b,c,x,y,z)$ for the stationary masses of
    $(E_n,E_\ell,E_t,X_n,X_\ell,X_t)$, respectively. Thus, in
    particular, $a=\pi(E_n)$. Reading the inflows from
    \eqref{eqn:block-transitions} gives
    \begin{align*}
        x&=\beta_X a,
        &
        y&=(1-\beta_X)a+b,
        &
        z&=c,\\
        a
        &=(1-\varepsilon)\beta_E x+\varepsilon y
            +\varepsilon\beta_E z,\\
        b
        &=(1-\varepsilon)(1-\beta_E)x+(1-\varepsilon)y
            +\varepsilon(1-\beta_E)z,\\
        c&=\varepsilon x+(1-\varepsilon)z.
    \end{align*}
    Since $z=c$, the last balance equation becomes $c=\varepsilon x+(1-\varepsilon)c$. Since $\varepsilon>0$, this implies $c=x$. Together with
    $x=\beta_Xa$, we obtain $c=x=z=\beta_Xa$. We now substitute these identities and
    $y=(1-\beta_X)a+b$ into the balance equation for $a$:
    \begin{align*}
        a =(1-\varepsilon)\beta_E(\beta_Xa) +\varepsilon\bigl((1-\beta_X)a+b\bigr) +\varepsilon\beta_E(\beta_Xa) =\beta_E\beta_Xa +\varepsilon(1-\beta_X)a +\varepsilon b.
    \end{align*}
    Rearranging gives
    \begin{align}
        \varepsilon b
        &=a\bigl[1-\beta_E\beta_X
            -\varepsilon(1-\beta_X)\bigr]
          =aH.
        \label{eqn:block-balance-H}
    \end{align}
    It remains to determine $a$ by normalization. Using
    $c=x=z=\beta_Xa$ and
    $y=(1-\beta_X)a+b$,
    \begin{align*}
        1=a+b+c+x+y+z  =a+b+3\beta_Xa+(1-\beta_X)a+b
        =2a(1+\beta_X)+2b.
    \end{align*}
    By \eqref{eqn:block-balance-H},
    $b=aH/\varepsilon$, and hence  $a={\varepsilon}/(2[\varepsilon(1+\beta_X)+H])$.
    Since $a=\pi(E_n)$, this is precisely the value denoted by
    $\alpha$ in the lemma statement. The remaining displayed
    stationary masses follow from
    \[
        x=c=z=\beta_Xa,\qquad
        b=\frac{aH}{\varepsilon},\qquad
        y=(1-\beta_X)a+b.
    \]

    A fault occurs on the next transition with probability one
    from $E_n$, with probability $1-\varepsilon$ from $X_n$, and with
    probability $\varepsilon$ from $X_t$. No other state contributes
    to the one-step fault probability. Consequently,
    \begin{align*}
        \rho(\beta_E,\beta_X)
        =a+(1-\varepsilon)x+\varepsilon z=a+(1-\varepsilon)\beta_Xa
            +\varepsilon\beta_Xa =a(1+\beta_X).
    \end{align*}
    Substituting the value of $a$ proves
    \eqref{eqn:block-general-rate}.

    Finally, suppose that
    $(\beta_E,\beta_X)=(0,0)$. In this case every state reaches the
    closed recurrent class $\{E_n,E_\ell,X_\ell\}$. Solving the
    balance equations on this class gives
    $\pi(E_n)=\varepsilon/2$. Since only $E_n$ contributes a fault,
    the stationary fault rate is $\varepsilon/2$.
\end{proof}

For $\mcA_\blkalg$, substitute
$\beta_E=1/2$ and
$\beta_X=1/[2(1-\varepsilon)]$ into
\eqref{eqn:block-general-rate}. This gives
\begin{align}
    \label{eqn:block-rate-ratio}
    \rho_{\blkalg}
    &=\varepsilon\left(1-\frac23\varepsilon\right),
    &
    \rho_{\latealg}
    &=\frac{\varepsilon}{2},
    &
    \frac{\rho_{\blkalg}}{\rho_{\latealg}}
    &=2-\frac43\varepsilon.
\end{align}

The six-state chain induced by $\mcA_\blkalg$ is irreducible.
The chain induced by $\mcA_\latealg$ has the unique recurrent class
identified in \Cref{lem:block-stationary}, and all its other states are
transient. Hence the finite-state Ces\`aro theorem applies from the
uniform initial page and common initial cache fixed in
\Cref{sec:prelim}, regardless of the period of the base request chain,
and yields
\[
    \E[\cost(\mcA_\blkalg,T)]
    =\rho_{\blkalg}T+o(T),
    \qquad
    \E[\cost(\mcA_\latealg,T)]
    =\rho_{\latealg}T+o(T).
\]
This proves \Cref{prop:block-rates}.

\end{document}